\documentclass{article}

\usepackage[preprint]{neurips_2021}

\usepackage[utf8]{inputenc} 
\usepackage[T1]{fontenc}    
\usepackage{hyperref}       
\usepackage{url}            
\usepackage{booktabs}       
\usepackage{amsfonts}       
\usepackage{nicefrac}       
\usepackage{microtype}      
\usepackage{xcolor}         

\newcommand{\p}[1]{\left( #1 \right)}
\newcommand{\br}[1]{\left[ #1 \right]}
\newcommand{\cb}[1]{\left\{ #1 \right\}}
\newcommand{\cd}[0]{\cdot}

\usepackage{amsmath}               
\usepackage{amsthm}                
\usepackage{amssymb} 
\usepackage{thmtools,thm-restate}
\usepackage[autostyle=true,german=quotes]{csquotes}
\usepackage{multirow}

\usepackage{array}
\usepackage{mathtools}
\newcolumntype{C}[1]{>{\centering\let\newline\\\arraybackslash\hspace{0pt}}m{#1}}

\newcommand{\preprint}{0}

\newtheorem{lemma}{Lemma}

\newtheorem{definition}{Definition}

\newcommand{\nplayer}[0]{\ensuremath{M}}
\newcommand{\gendist}[0]{\ensuremath{\Theta}}
\newcommand{\mean}[0]{\ensuremath{\theta}}
\newcommand{\mue}[0]{\ensuremath{\mu_e}}
\newcommand{\var}[0]{\ensuremath{\sigma^2}}
\newcommand{\sampledist}[0]{\ensuremath{\mathcal{D}}}

\newcommand{\ndraw}[0]{\ensuremath{n}}
\newcommand{\err}[0]{\ensuremath{\epsilon^2}}

\newcommand{\expparam}[1]{\ensuremath{\mathbb{E}_{(\mean_{#1}, \err_{#1}) \sim \gendist}}}
\newcommand{\total}[0]{\ensuremath{N}}

\newcommand{\Ymf}[0]{\ensuremath{Y}}

\newcommand{\alone}[0]{\ensuremath{\pi_l}}
\newcommand{\gcol}[0]{\ensuremath{\pi_g}}

\newcommand{\s}[0]{\ensuremath{s}}
\newcommand{\el}[0]{\ensuremath{\ell}}
\newcommand{\ns}[0]{\ensuremath{n_{\s}}}
\newcommand{\nlv}[0]{\ensuremath{n_{\el}}}

\newcommand{\col}[0]{\ensuremath{C}}
\newcommand{\colA}[0]{\ensuremath{A}}

\newcommand{\costw}[0]{\ensuremath{f_w}}
\newcommand{\partition}[0]{\ensuremath{\Pi}}

\title{Optimality and Stability in Federated Learning: \\A Game-theoretic Approach}

\author{%
  Kate Donahue \\
  Department of Computer Science\\
  Cornell University\\
  \texttt{kdonahue@cs.cornell.edu} \\
   \And
   Jon Kleinberg \\
   Departments of Computer Science \\and Information Science \\
   Cornell University\\
   \texttt{kleinber@cs.cornell.edu} \\
}

\begin{document}

\maketitle

\begin{abstract}
Federated learning is a distributed learning paradigm where multiple agents, each only with access to local data, jointly learn a global model. There has recently been an explosion of research aiming not only to improve the accuracy rates of federated learning, but also provide certain guarantees around social good properties such as total error. One branch of this research has taken a game-theoretic approach, and in particular, prior work has viewed federated learning as a hedonic game, where error-minimizing players arrange themselves into federating coalitions. This past work proves the existence of stable coalition partitions, but leaves open a wide range of questions, including how far from optimal these stable solutions are. In this work, we motivate and define a notion of optimality given by the average error rates among federating agents (players). First, we provide and prove the correctness of an efficient algorithm to calculate an optimal (error minimizing) arrangement of players. Next, we analyze the relationship between the stability and optimality of an arrangement. First, we show that for some regions of parameter space, all stable arrangements are optimal (Price of Anarchy equal to 1). However, we show this is not true for all settings: there exist examples of stable arrangements with higher cost than optimal (Price of Anarchy greater than 1). Finally, we give the first constant-factor bound on the performance gap between stability and optimality, proving that the total error of the worst stable solution can be no higher than 9 times the total error of an optimal solution (Price of Anarchy bound of 9).
\end{abstract}

\section{Introduction}

Recent advances of machine learning techniques has made it possible to apply powerful prediction algorithms to a variety of domains. However, in real-world situations, data is often distributed across multiple locations and cannot be combined to a central repository for training. For example, consider patient medical data located at hospitals or student educational data at different schools. In each case, the individual agents (hospitals or schools) who hold the data wish to find a model that minimizes their error. However, the data at each location may be insufficient to train a robust model. Instead, the agents may prefer to build a model using data from multiple agents: multiple hospitals or schools. Collectively, the combined data may be able to produce a model with much higher accuracy, providing more powerful predictions to each agent and increasing overall welfare. However, it may be infeasible to transfer the data to some coordinating entity to build a global model: privacy, data size, and data format are all possible reasons that would make transferring data not a reasonable solution. 

Federated learning is a novel distributed learning paradigm that aims to solve this problem (\cite{mcmahan2016communicationefficient}). Data remains at separate local sites, which individual agents use to learn local model parameters or parameter updates. Then, only the parameters are transferred to the coordinating entity (for example, a technology company), which averages together all of the parameters in order to form a single global model, which all of the agents use. Federated learning is a rapidly growing area of research (\cite{Li_2020, kairouz2019advances, fedsurvey}). 

However, research has also noted that federated learning, in its traditional form, may not be the best option for each agent (\cite{yu2020salvaging, bagdasaryan2019differential, li2019fair, mohri2019agnostic}). In the real world, agents may differ in their true distribution: the true model of patient outcomes at hospital $A$ may differ from the true model at hospital $B$, for example. If these differences are large enough, federating agents may see their error increase under certain situations, potentially even beyond what they would have obtained with only local learning. For example, a player with relatively few samples may end up seeing its model \enquote{torqued} by the presence of a player with many samples. For this reason, agents may not wish to federate with every other potential agent. 

Instead, each agent faces a choice: given the costs and benefits of federating with different players, it must determine which of the exponentially many combinations of players it would prefer to federate with. Simultaneously, every other agent is also attempting to identify and join a federating group that it prefers - and agents may have conflicting preferences. Prior work (\cite{donahue2020model, hasan2021incentive}) has formulated this problem as a \emph{hedonic game}, which each player derives some cost (error) from the coalition they join. The aim of such research has been to identify partitions of players that are \emph{stable} against deviations, for varying definitions of stability. A hedonic game in general may not have any stable arrangements, so the area's contributions in the analysis of stability adds valuable insight into the incentives of federating agents.

However, this framework also leaves open multiple game theoretic questions. While the federating agents have individual incentives to reduce their error, society as a whole also has an interest in minimizing the overall error. In the school example, individual schools wish to find coalitions that work well on their own sub-populations, while the overall district or state may have an interest in finding an overall set of coalitions that minimizes the overall error. This analysis of a coalition partition's overall cost falls under the game theoretic notion of \emph{optimality}.

One natural question relates to the tension between these two goals: the self-interested goal of the individual actors (stability) and the overall goal of reducing total cost (optimality). Given a that set of self-interested agents has found a stable solution, how far from optimal could it be? This is reflected by the \emph{Price of Anarchy} of a game, the canonical approach to study optimality and stability jointly (\cite{papadimitriou2001algorithms, koutsoupias1999worst}). The Price of Anarchy (PoA) is a ratio where the numerator is equal to the highest-cost stable arrangement and the denominator is equal to the lowest-cost arrangement (the optimal arrangement). It is lower bounded by 1, a bound that it achieves only if all stable arrangements are optimal. A higher Price of Anarchy value implies a greater trade off between stability and optimality, and bounding the Price of Anarchy for a particular game puts a limit on this trade-off. Federated learning is a situation where questions of stability have been analyzed, but to our knowledge there has been no systematic analysis of the Price of Anarchy in a model of federated learning.

\paragraph{\bf The present work: A framework for optimality and stability in federated learning} 

In this work, we make two main contributions to address this gap. First, we provide an efficient, constructive algorithm for calculating an optimal federating arrangement. Secondly, we prove the first-ever constant bound on the Price of Anarchy for this game, showing that the worst stable arrangement is no more than 9 times the cost of the best arrangement.

We begin Section \ref{sec:opt} by defining optimality, drawing on a notion of weighted error derived from the standard objective in federated learning literature. The main contribution of this section is an efficient, constructive algorithm for calculating an optimal arrangement, along with a proof of its optimality. However, as demonstrated in Section \ref{sec:PoA}, optimality and stability are not always simultaneously achieved. This section analyzes the Price of Anarchy, which measures how far from optimal the worst stable arrangement can be. First, we demonstrate that the optimal arrangement is not always stable. Next, we show that there exist sub-regions where the Price of Anarchy is equal to 1. Finally, this section proves an overall Price of Anarchy bound of 9, the first constant bound for this game.

It is worth emphasizing that, beyond the Price of Anarchy bound itself, part of the contribution of this work is the optimization and analysis to produce this bound. The proofs for this contribution are modular and illuminate multiple properties about the broader federated learning game under study. As such, these contributions could be useful for further investigating this model. For example, the modular structure of our proof is what enables us to establish stronger bounds for certain sub-cases. 


\section{Related work}\label{sec:relatedwork}

\paragraph{\bf Federated learning}

As we mentioned previously, federated learning has recently seen numerous advances. In this section, we highlight a few papers in federated learning that are especially related to our work.

The idea that agents might differ in their true models (that data might be generated non-i.i.d. across multiple agents) is commonly acknowledged in the federated learning literature. For example, \cite{yu2020salvaging, bagdasaryan2019differential} empirically demonstrate that federated learning, and especially privacy-related additions, can cause a wide disparity in error rates. Some techniques have been developed specifically towards this problem.  For example, hierarchical federated learning adds an additional layer of hierarchical structure to federated learning, which could be used to reduce latency or to cluster together similar players (\cite{lin2018dont, Liu_2020}). Many other works also relate to clustering, such as (\cite{Lee2020AccurateAF, Sattler_2020, ShlezingerClustFed, wang2020split, duanfedgroup, jamali2021federated, caldarola2021cluster}). These works, which tend to be more applied than our work, may also differ in that they analyze situations where additional information is known, such as the data distribution at each location. 

Other work aims to improve accuracy rates by selecting acquiring additional data (\cite{collabPAC, Duan_selfbalance}). Some papers specifically analyze federated learning for high-stakes situations such as medical settings (\cite{xia2021auto, guo2021multi, vaid2021federated, kumar2021blockchain, zhang2021dynamic}). In general, all of these works have the goal of reducing the average error over all federating agents, which we will use to motivate our definition of optimality in later sections. 

\paragraph{\bf Game theory in federated learning}

The closest work to this current paper is \citet{donahue2020model}, which we discuss in greater detail in Section \ref{sec:model}. Another paper using hedonic game theory to analyze federated learning games is \citet{hasan2021incentive},  which gives conditions for Nash stability in federated learning. Other works analyze the incentives of players to contribute resources towards federated learning: \citet{blum2021one} analyzes fairness and efficiency in sampling additional points for federated learning and \citet{incentivemech} analyzes incentives for agents to contribute computational resources in federated learning. Interestingly, multiple works take a game theoretic approach towards coalition formation in cloud computing, but with the aim of minimizing some cost besides error, such as electricity usage \cite{Guazzone_2014, Anglano2018AGA}. 

\section{Model and assumptions}\label{sec:model}

We assume that there are $\nplayer$ total agents (sometimes referred to as players). Each agent $i \in [\nplayer]$ has drawn $\ndraw_i$ data points from their true local distribution $g(\mean_i)$, where $\mean_i$ are their true local parameters and $g(\cd)$ is some true labeling function. A player's goal is to learn a model with low expected error on its own distribution. If a player opts for local learning, then it uses its local estimate of these parameters $g(\hat \mean_i)$ to predict future data points, obtaining error $err_i(\{i\})$. If a set of players $\col$ are federating together, we say that they are in a \emph{coalition}  or \emph{cluster} together. They combine their local estimates of parameters into a single federated estimate, governed by the weighted average of their parameters:  
\begin{equation}\label{eq:avged}
 \hat \mean_\col = \frac{1}{\sum_{i \in \col} \ndraw_i} \cd \sum_{i \in \col} \ndraw_i \cd \hat \mean_i  
\end{equation}
A federating player $i$ obtains error $err_i(\col)$: note that this value may differ between players in the same coalition. For example, if player $j$ has samples than player $k$, then $\hat \mean_\col$ will be weighted more towards player $j$, meaning that player $j$ will have lower expected error than $k$. 

The weighted average method in Equation \ref{eq:avged} is commonly used in federated learning (\cite{mcmahan2016communicationefficient}). Because it is the most straightforward method, it is sometimes called \enquote{vanilla} federated learning. Alternative ways of federation might involve customizing the model for individuals, as in domain adaptation. For example, \cite{donahue2020model} models three methods of federation: vanilla (called \enquote{uniform}), as well as two models of domain adaptation. 

There are multiple reasons why we opted to analyze the federation method in Equation \ref{eq:avged} in this work. First of all, this federation method is the most straightforward method, and as such it is the natural candidate to begin analysis. Secondly, this federation method is the most interesting to analyze technically. Domain adaptation serves to increase the incentives of an individual agent to participate in federation: it reduces the tension between an individual's incentives and the optimal overall arrangement. Because of this, for Price of Anarchy it is more valuable and challenging to explore the case in Equation \ref{eq:avged}, where incentives are more opposed. 


\subsection{Theoretical model of federation from \cite{donahue2020model}}

Federated learning has been the subject of both applied and theoretical analysis; our focus here is on the theoretical side. 
In addition, for game theoretical reasoning to be feasible, we need a model that gives exact errors (costs) for each player, rather than bounds: these are needed in order to be able to argue that a certain arrangement is optimal, for example. 



\ifthenelse{\equal{\preprint}{1}}{
We opt to use the model developed by \citet{donahue2020model}, which produces the closed-form error value seen in Lemma \ref{lem:error} below. While we work within this model, we emphasize that \citet{donahue2020model} asks very different questions from our area of focus: they primarily analyze the stability of federated learning coalitions, while we analyze optimality and Price of Anarchy. 
\begin{lemma}[Lemma 4.2, from \cite{donahue2020model}]\label{lem:error}
Consider a mean estimation task as follows: player $j$ is trying to learn its true mean $\mean_j$. It has access to $\ndraw_j$ samples drawn i.i.d. $\Ymf \sim \sampledist_j(\mean_j, \err_j)$, a distribution with mean $\mean_j$ and variance $\err_j$. Given a population of players, each has drawn parameters $(\mean_j, \err_j) \sim \Theta$ from some common distribution $\Theta$. A coalition $\col$ federating together produces a single model based on the weighted average of local means (Eq. \ref{eq:avged}). Then, the expected mean squared error player $j$ experiences in coalition $\col$ is:
\begin{equation}\label{eq:err}
err_j(\col) = \frac{\mue}{\sum_{i \in \col} \ndraw_i} + \var \cd \frac{\sum_{i \in \col, i \ne j}\ndraw_i^2 + \p{\sum_{i\in \col, i \ne j}\ndraw_i}^2}{\p{\sum_{i \in \col}\ndraw_i}^2}
\end{equation}
where $\mue = \expparam{i}[\err_i]$ (the average noise in data sampling) and $\var = Var(\mean_i)$ (the average distance between the true means of players). 
\end{lemma}
Note that \cite{donahue2020model} also analyzes a linear regression game with a similar cost function, though in this work we will focus on the mean estimation game.

We use some of the same notion and modeling assumptions as \cite{donahue2020model}. For example, we use $\col$ to refer to a coalition of federating agents and $\partition$ to refer to a collection of coalitions that partitions the $\nplayer$ agents. We will use $\total_{\col}$ to refer to the total number of samples present in coalition $\col$: $\total_{\col} = \sum_{i \in \col}\ndraw_i$. In a few lemmas we will re-use minor results proven in \cite{donahue2020model}, citing them for completeness. 

For technical assumptions, we assume number of samples $\{\ndraw_i\}$ is fixed and known by all. We also assume that the parameters $\mue, \var$ are approximately known: in particular, results will depend on whether the number of samples is larger or smaller than the critical threshold $\frac{\mue}{\var}$. We assume that a player does not know anything else about its own true parameters $\mean_i$ or the parameters of other players: for example, it does not know the true generating distribution $\Theta$ or if its true parameters are likely to lie far from the parameters of other players. We assume that each player has a goal of obtaining a model with low expected test error on its personal distribution, and that the federating coordinator is motivated to minimize some notion of total cost, but is otherwise impartial. 

Finally, it is worth emphasizing key differences between our work and \cite{donahue2020model}. The focus of \cite{donahue2020model} is defining a theoretical model of federated learning and analyzing the stability of such an arrangement. As such, it focuses solely on individual incentives and completely omitted any analysis of overall societal welfare. Our work focuses on discussions of optimality (overall welfare) and Price of Anarchy, questions that are completely distinct from their prior work. Additionally, our work is in some ways more general: while some key results in \cite{donahue2020model} only allow players to have two different numbers of samples (\enquote{small} or \enquote{large}), every result in our work holds for arbitrarily many players with arbitrarily many different numbers of samples.}{

We opt to use the model developed in our prior work \citet{donahue2020model}, which produces the closed-form error value seen in Lemma \ref{lem:error} below. While we work within this model, we emphasize that \citet{donahue2020model} asked different questions from this paper's focus: our prior work focused on developing the federated learning model and analyzing the stability of federating coalitions, while our current work analyzes optimality and Price of Anarchy. 
\begin{lemma}[Lemma 4.2, from \cite{donahue2020model}]\label{lem:error}
Consider a mean estimation task as follows: player $j$ is trying to learn its true mean $\mean_j$. It has access to $\ndraw_j$ samples drawn i.i.d. $\Ymf \sim \sampledist_j(\mean_j, \err_j)$, a distribution with mean $\mean_j$ and variance $\err_j$. Given a population of players, each has drawn parameters $(\mean_j, \err_j) \sim \Theta$ from some common distribution $\Theta$. A coalition $\col$ federating together produces a single model based on the weighted average of local means (Eq. \ref{eq:avged}). Then, the expected mean squared error player $j$ experiences in coalition $\col$ is:
\begin{equation}\label{eq:err}
err_j(\col) = \frac{\mue}{\sum_{i \in \col} \ndraw_i} + \var \cd \frac{\sum_{i \in \col, i \ne j}\ndraw_i^2 + \p{\sum_{i\in \col, i \ne j}\ndraw_i}^2}{\p{\sum_{i \in \col}\ndraw_i}^2}
\end{equation}
where $\mue = \expparam{i}[\err_i]$ (the average noise in data sampling) and $\var = Var(\mean_i)$ (the average distance between the true means of players). 
\end{lemma}
Note that \cite{donahue2020model} also analyzes a linear regression game with a similar cost function, though in this work we will restrict our attention to the mean estimation game.

We use some of the same notion and modeling assumptions as \cite{donahue2020model}. For example, we use $\col$ to refer to a coalition of federating agents and $\partition$ to refer to a collection of coalitions that partitions the $\nplayer$ agents. We will use $\total_{\col}$ to refer to the total number of samples present in coalition $\col$: $\total_{\col} = \sum_{i \in \col}\ndraw_i$. In a few lemmas we will re-use minor results proven in \cite{donahue2020model}, citing them for completeness. 

For technical assumptions, we assume number of samples $\{\ndraw_i\}$ is fixed and known by all. We also assume that the parameters $\mue, \var$ are approximately known: in particular, results will depend on whether the number of samples is larger or smaller than the critical threshold $\frac{\mue}{\var}$. We assume that a player does not know anything else about its own true parameters $\mean_i$ or the parameters of other players: for example, it does not know the true generating distribution $\Theta$ or if its true parameters are likely to lie far from the parameters of other players. We assume that each player has a goal of obtaining a model with low expected test error on its personal distribution, and that the federating coordinator is motivated to minimize some notion of total cost, but is otherwise impartial. 

Finally, it is worth emphasizing key differences between this current work and \cite{donahue2020model}. The focus of \cite{donahue2020model} is defining a theoretical model of federated learning and analyzing the stability of such an arrangement. As such, it focuses solely on individual incentives, rather than overall societal welfare. On this other hand, this current work focuses on discussions of optimality (overall welfare) and Price of Anarchy. Finally, this paper work is in some ways more general: while some results in \cite{donahue2020model} only allow players to have two different numbers of samples (\enquote{small} or \enquote{large}), every result in our work holds for arbitrarily many players with arbitrarily many different numbers of samples. This distinction is a function of the questions analyzed in each paper: questions of stability (as in \cite{donahue2020model}) are much harder to analyze for players with arbitrarily many different sizes. 

}

\section{Optimality}\label{sec:opt}
We will begin with the question of optimality. As motivation, it is useful to consider the objective function of most federated learning papers \cite{mcmahan2016communicationefficient}: 
\begin{equation*}
    \min_\mean err_w(\mean) = \sum_{i=1}^{\nplayer}p_i \cd err_i(\mean) =^* \frac{1}{\sum_{i=1}^{\nplayer}\ndraw_i}\sum_{i=1}^{\nplayer}\ndraw_i\cd err_i(\mean)
\end{equation*}
While the weights can be any $p_i > 0, \sum_{i=1}^{\nplayer}p_i=1$, the $*$ equality reflects the common setting where they are taken to be the empirical average. In this work, we will take the empirical average as our cost function: 
\begin{definition}\label{def:opt}
A coalition partition $\partition$ is optimal if it minimizes the weighted sum of errors across players, as defined below: 
$$\costw(\Pi) = \sum_{\col \in \partition}\costw(\col) = \sum_{\col \in \partition}\sum_{i \in \col}\ndraw_i \cd err_{i}(\col)$$
We will say that a coalition partition $\partition$ is in $OPT$ if it achieves minimal cost. Note that multiple partitions may achieve minimal cost, so $OPT$ is a set of partitions. 
\end{definition}
Because $\partition$ is a disjoint partition over the $\nplayer$ players, $\costw(\partition)$ is simply the error $err_w(\mean)$ scaled by a constant. Therefore, minimizing $\costw(\partition)$ is equivalent to minimizing the weighted average of errors.

 
Some machine learning papers modify the empirical average objective to achieve other goals. For example, \citet{li2019fair, mohri2019agnostic} consider variants where this goal is re-weighted in order to achieve certain fairness goals. Appendix \ref{app:otherdef} discusses other possible cost functions. 

All of the above analysis holds for any model of federated learning. Lemma \ref{lem:optdef}, below, gives the specific form of cost for federated learning using the model from \cite{donahue2020model}. The remaining analysis in this paper will assume this cost function. Proofs for results in this section are given in Appendix \ref{app:opt}.

\begin{restatable}{lemma}{optdef}
\label{lem:optdef}
Consider a partition $\partition$ made up of coalitions $\{\col_i\}$. Then, using the error form given in Equation \ref{eq:err}, the total cost of $\partition$ is given by
$$\costw(\partition) = \sum_{\col \in \partition}\cb{\mue + \var \cd \total_{\col} - \var \frac{\sum_{i \in \col} \ndraw_i^2}{\total_{\col}}}$$
\end{restatable}

The two most common arrangements in machine learning tasks are local learning (which we will denote by $\alone$) and the federation in the \emph{grand coalition} ($\gcol$), where all of the players are federating together in a single coalition. However, Lemmas \ref{lem:alonebad} and \ref{lem:gcolbad} demonstrate that either of these could could perform arbitrarily poorly as compared the cost-minimizing (optimal) arrangement.

\begin{restatable}{lemma}{alonebad}
\label{lem:alonebad}
$\forall \rho>1$, there exists a setting where local learning results in average error more than $\rho$ times higher than optimal: $\frac{\costw(\alone)}{\costw(OPT)} > \rho$. 
\end{restatable}

\begin{restatable}{lemma}{gcolbad}
\label{lem:gcolbad}
$\forall \rho > 1$, there exists a setting where federating in the grand coalition results in average error more than $\rho$ times higher than optimal: $\frac{\costw(\gcol)}{\costw(OPT)} > \rho$. 
\end{restatable}

A priori, finding a partition of players that minimizes total cost seems extremely challenging. There are exponentially many options for partitions, and two lemmas above have shown that either of the most common choices could be arbitrarily far from optimal. However, next section will provide an efficient, constructive algorithm to calculate an optimal partition of players into federating coalitions.

\subsection{Calculating an optimal arrangement}

The main contribution of this section is Theorem \ref{thrm:optcalc} gives an algorithm for minimizing the total weighted error of the federating agents. 

\begin{restatable}{theorem}{optcalc}
\label{thrm:optcalc}
Consider a set of players $\{\ndraw_i\}$. An optimal partition $\partition$ can be created as follows: first, start with every player doing local learning. Then, begin by grouping the players together in ascending order of size, stopping when the first player would increase its error by joining the coalition from local learning. Then, the resulting partition $\partition$ is optimal.
\end{restatable}

Though the algorithm in Theorem \ref{thrm:optcalc} is straightforward, proving the optimality of the resulting partition $\partition$ requires several sub-lemmas. Each sub-lemma is a building-block that describes certain operations that either increase or decrease total cost. The proof of Theorem \ref{thrm:optcalc} largely consists of sequentially using these sub-lemmas in order to demonstrate the optimality of the calculated partition.

\paragraph{\bf Statement and description of supporting lemmas} First, Lemma \ref{lem:addminsame} demonstrates a close relationship between movements of players that reduce total cost and movements of players that are in that player's self-interest (recall that players always wish to minimize their expected error). Specifically, it shows that a player wishes to join a coalition from local learning if and only if that move would reduce total cost for the entire partition. 

\begin{restatable}[Equivalence of player preference and reducing cost]{lemma}{addminsame}
\label{lem:addminsame}
Take any coalition $Q$ and any player $j$. Then, a player wishes to join that coalition (from local learning) if and only if doing so would reduce total cost. That is, 
$$\costw(\{\ndraw_j\}) + \costw(Q) \geq \costw(\{\ndraw_j\} \cup Q) \quad \Leftrightarrow \quad err_j(\{\ndraw_j\}) \geq err_j(\{\ndraw_j\} \cup Q)$$
\end{restatable}

Next, Lemma \ref{lem:swap} shows that \enquote{swapping} the roles of two players (one doing local learning, one federating in a coalition) reduces total cost when the larger player is removed to local learning. 

\begin{restatable}[Swapping]{lemma}{swap}
\label{lem:swap}
Take any set $Q$ including a player $\ndraw_j> \ndraw_k$, where the player $\ndraw_k$ is doing local learning. Then, swapping the roles of players $k$ and $j$ always decreases total cost. 
$$\costw(Q \cup \{\ndraw_j\}) + \costw(\{\ndraw_k\}) >\costw(Q \cup \{\ndraw_k\}) + \costw(\{\ndraw_j\})  $$
\end{restatable}

Lemmas \ref{lem:orderjoin} and \ref{lem:wontleave} give results for when players are incentivized to leave or join a particular coalition: they show that such incentives are monotonic in the size of the player. By Lemma \ref{lem:addminsame}, these results also show the monotonicity of cost-reducing operations. Note that these lemmas are not equivalent: they differ in whether the reference player $j$ is already in the coalition or not. 

\begin{restatable}[Monotonicity of joining]{lemma}{orderjoin}
\label{lem:orderjoin}
If a player of size $\ndraw_j$ would prefer local learning to joining a coalition $Q$, then any player of size $\ndraw_k \geq \ndraw_j$ also prefers local learning to joining the same coalition. That is, for $\ndraw_k \geq \ndraw_j$,
$$err_j(Q\cup \{\ndraw_j\}) \geq err_j(\{\ndraw_j\}) \quad \Rightarrow \quad err_k(Q \cup \{\ndraw_k\}) \geq err_k(\{\ndraw_k\}) $$
Conversely, if a player $j$ wishes to join $Q$, then any other player of size $\ndraw_k \leq \ndraw_j$ would have also wanted to join. That is, for $\ndraw_j \geq \ndraw_k$, 
$$err_j(Q\cup \{\ndraw_j\}) \leq err_j(\{\ndraw_j\}) \quad \Rightarrow \quad err_k(Q \cup \{\ndraw_k\}) \leq err_k(\{\ndraw_k\}) $$
\end{restatable}

\begin{restatable}[Monotonicity of leaving]{lemma}{wontleave}
\label{lem:wontleave}
Take any coalition $Q$. Then, if any player $j \in Q$ of size $\ndraw_j$ wishes to leave $Q$ for local learning, then any player of size $\ndraw_k \geq \ndraw_j$ also wishes to leave for local learning. That is, for $\ndraw_k \geq \ndraw_j$
$$err_j(Q) \geq err_j(\{\ndraw_j\}) \quad \Rightarrow \quad err_k(Q) \geq err_k(\{\ndraw_k\})$$
Conversely, if a player $j \in Q$ of size $\ndraw_j$ does \emph{not} wish to leave $Q$ for local learning, then any player $k \in Q$ of size $\ndraw_k \leq \ndraw_j$ also does not wish to leave. That is, for $\ndraw_k \leq \ndraw_j$
$$err_j(Q) \leq err_j(\{\ndraw_j\}) \quad \Rightarrow \quad err_k(Q) \leq err_k(\{\ndraw_k\})$$
\end{restatable}

All of the above lemmas have analyzed situations where a single player is moving between coalitions. Lemma \ref{lem:wontleave} analyzes cases where multiple players are rearranged simultaneously. Specifically, it provides an algorithm for combining together two separate groups (and then removing certain players) that is guaranteed to keep constant or reduce total cost.

\begin{restatable}[Merging]{lemma}{merge}
\label{lem:merge}
Consider two groups of players, $P, Q$. First, merge together the two groups to form $P \cup Q$. Then, remove players from $P\cup Q$ to local learning, removing them in descending order of size. Stop removing players when the first player would prefer to stay (removing it would increase its error). Then, this overall process maintains or decreases total error. 
In other words, 
\begin{equation}\label{eq:submod0}
\costw(Q) + \costw(P) \geq \costw(\{Q\cup P\}\setminus L) + \sum_{i \in L}\costw(\{\ndraw_i\})
\end{equation}
where $L$ is the set of large players removed in descending order of size. The inequality is strict so long as the final structure is not identical to the first, up to renaming of players, and it is \emph{not} the case that all the players have the exact same size. 
\end{restatable}

The proof of Theorem \ref{thrm:optcalc} is given simply by applying the lemmas sequentially to show that any other partition $\partition'$ can be converted to the described partition $\partition$ through a series of operations that decrease or hold constant total cost. 

\section{Price of Anarchy}\label{sec:PoA}

\begin{table}[]
\centering 
\begin{tabular}{|c|c|c|c|c|c|}
\hline
Coalition structure    & $err_a(\cd), \ndraw_a =1$ & $err_b(\cd), \ndraw_b = 8$ & $err_c(\cd), \ndraw_c=15$ & $\costw(\partition)$ & $err_w(\partition)$ \\ \hline
$\{a,\}, \{b\}, \{c\}$ & 10                        & 1.25                       & 0.667                     & 30                   & 1.25                \\ \hline
$\{a\}, \{b, c\}$      & 10                        & 1.285                      & 0.677                     & 30.435               & 1.268               \\ \hline
$\{a, c\},\{b\}$       & 2.382                     & 1.25                       & 0.633                     & 21.875               & 0.911               \\ \hline
$\{a, b\}, \{c\}$       & 2.691                     & 1.136                      & 0.667                     & 21.778               & 0.907               \\ \hline
$\{a, b, c\}$          & 1.834                     & 1.253                      & 0.670                     & 21.917               & 0.913               \\ \hline
\end{tabular}
\caption{Example with $\mue =10, \var =1$ example with three players of size $\ndraw_a = 1, \ndraw_b = 8, \ndraw_c = 15$. Note that $\{a, b\}, \{c\}$ minimizes total cost, but is not individually stable: player $a$ wishes to leave its coalition to join player $c$, which welcomes that player joining it. This produces $\{a, c\}, \{b\}$, which is the only individually stable arrangement, giving a Price of Anarchy value of $21.875/21.778 = 1.0045$. }
\label{tab:case21}
\end{table}

The previous section defined the \enquote{optimality} of a federating arrangement as its average error, and additionally provided an efficient algorithm to calculate a lowest-cost arrangement. Given that much of prior work (\cite{donahue2020model, hasan2021incentive}) has studied the stability of cooperative games induced by federated learning, the next natural question is to study the relationship between stability and optimality. This section analyzes this relationship, using the canonical game theoretic tools of Price of Anarchy and Price of Stability. All proofs for this section are in Appendix \ref{app:poa}. 

First, we will define the notions of stability under analysis, which are all drawn from standard cooperative game theory literature (\cite{BOGOMOLNAIA2002201}). A partition of players $\partition$ is \emph{core stable} if there does not exist a set of players that all would prefer leave their location in $\partition$ and form a coalition together. A partition is \emph{individually stable} (IS) if there does not exist a single player $i$ that wishes to join some existing coalition $\col$, where all members of $\col$ weakly prefer that $i$ join. Our results will primarily use the notion of individual stability. 

As a reminder, the Price of Anarchy (PoA) is the ratio between the worst (highest-cost) stable arrangement and the best (lowest-cost) arrangement. The Price of Stability is the ratio of the best stable arrangement and the best overall arrangement (regardless of if it is stable or not) (\cite{anshelevich2008price}). Note that the Price of Stability is 1 when there exists an optimal arrangement that is also stable. 

First, we will show that for certain ranges of parameter space, the Price of Anarchy and/or Price of Stability are equal to 1. Specifically, Lemma \ref{lem:gcolcore} shows that when all players have relatively few samples (no more than $\frac{\mue}{\var}$ each), the grand coalition $\gcol$ is core stable, implying a Price of (Core) Stability of 1. Recall that $\mue$ and $\var$ are parameters of the federated learning model reflecting the average noise of the data and the average dissimilarity between federating agents, respectively. 

\begin{restatable}{lemma}{gcolcore}
\label{lem:gcolcore}
For a set of players with $\ndraw_i \leq \frac{\mue}{\var} \ \forall i$, the grand coalition $\gcol$ is always core stable. 
\end{restatable}

On the other hand, Lemma \ref{lem:alonecore} shows that  when all players have relatively many samples (at least $ \frac{\mue}{\var}$ each), every core or individually stable arrangement is also optimal, which means that the Price of Anarchy for this situation is 1.  

\begin{restatable}{lemma}{alonecore}
\label{lem:alonecore}
For a set of players with $\ndraw_i \geq \frac{\mue}{\var} \ \forall i$, any arrangement that is core stable or individually stable is also optimal. 
\end{restatable}

However, it is \emph{not} the case that either the Price of Stability or Price of Anarchy is always 1. Table \ref{tab:case21} contains an example demonstrating this: there exists a simple three-player case where the optimal arrangement is not individually stable. However, the Price of Anarchy value here is quite small, which suggests the prospect that the Price of Anarchy in general could be bounded. 

The main result of this section is Theorem \ref{thrm:PoA}, which proves a Price of Anarchy bound of 9 for this problem: the cost of the highest stable arrangement is no more than 9 times the cost of the optimal (lowest cost) arrangement. 

\begin{restatable}[Price of Anarchy]{theorem}{PoA}
\label{thrm:PoA}
Denote $\partition_M$ to be a maximum-cost individually stable (IS) partition and $\partition_{opt}$ to be an optimal (lowest-cost) partition. Then, 
$$PoA = \frac{\costw(\partition_M)}{\costw(\partition_{opt})} \leq 9$$
\end{restatable}

In Theorem \ref{thrm:PoA}, the numerator is the cost of $\partition_M$, a maximum-cost partition, and the denominator is $\partition_{opt}$, an optimal (lowest-cost) partition. Recall that Definition \ref{def:opt} gives the cost of an arrangement as the weighted sum of the errors of the respective players. Therefore, to get an upper bound on the Price of Anarchy, we will upper bound the errors players experience in $\partition_M$ and lower bound on the error players experience in $\partition_{opt}$. 

\paragraph{\bf Summary of proof technique} Again, this section will show how the larger theorem is the result of several lemmas that act as building blocks. In particular, the lemmas will take two separate approaches towards creating the bound. Lemmas of the first type (\ref{lem:betterthanalone}, \ref{lem:lowerbounderror}, \ref{lem:twcspec}) all provide upper or lower bounds on the errors certain players can experience. These conditions depend on the size of the player (how many samples it has) and the size of the group it is federating with (how many samples in total the rest of the coalition has). For example, Lemmas \ref{lem:betterthanalone} and \ref{lem:lowerbounderror} taken together show that a player with at least $\frac{\mue+ \var}{2\var}$ samples has a worst-case error no more than 2 times its best-case error. The same pair of lemmas give a multiplicative bound of 9 for players with numbers of samples that falls between $\frac{\mue}{9\cd \var}$ and $\frac{\mue+ \var}{2\var}$. Finally, Lemmas \ref{lem:twcspec} and \ref{lem:lowerbounderror} together give a factor of 7.5 for players with fewer than $\frac{\mue}{9\cd \var}$ samples that are federating with other players of total size at least $\frac{\mue}{3\cd \var}$. Taken together, these errors show that, for almost all cases, the highest error a player experiences is no more than 9 times higher than the lowest error it might experience. 

The final case that needs to be addressed is when a player of size $\leq \frac{\mue}{9\cd \var}$ is federating in a group with other players of total size $\leq \frac{\mue}{3\cd \var}$. Lemma \ref{lem:relaxed} handles this last case by an argument around stability. Specifically, it shows that any players in such an arrangement can only be stable if all of them are grouped together into a single federating coalition. In the proof of Theorem \ref{thrm:PoA}, this result ends up enabling an additive factor to the Price of Anarchy bound, which is absorbed into the other factors for a total Price of Anarchy value of 9. 

\paragraph{\bf Statement and description of supporting lemmas} Next, we will walk through each lemma specifically. Lemma \ref{lem:betterthanalone} gives an \emph{upper} bound of $\frac{\mue}{\ndraw_i}$ on the error any player experiences in $\partition_M$. 

\begin{lemma}\label{lem:betterthanalone}
If $\partition_M$ is a maximum-cost IS partition, then $err_i(\partition_M) \leq \frac{\mue}{\ndraw_i}$ for all players $i$.
\end{lemma}
\begin{proof}
Because $\partition_M$ is individually stable, every player must get error no more than the error it would receive alone (doing local learning). By Lemma \ref{lem:error} with $C= \ndraw_i$, a player with samples $\ndraw_i$ player gets error $\frac{\mue}{\ndraw_i}$ alone.  
\end{proof}

Next, Lemma \ref{lem:lowerbounderror} provides \emph{lower} bounds on the error a player can receive in $\partition_{opt}$. It does this by bounding the minimum error a player could get in any arrangement. Again, because the cost of $\partition_{opt}$ is simply the weighted sum of errors of each individual player, this helps to upper bound the Price of Anarchy. First, Lemma \ref{lem:lowerbounderror} shows that for players with at least $\frac{\mue+\var}{2\var}$ samples, the lowest possible error it could experience is $\frac{1}{2}\cd \frac{\mue}{\ndraw_j}$, which is a factor of 2 off from its worst-case error in Lemma \ref{lem:betterthanalone}. For players with fewer samples than $\frac{\mue+\var}{2\var}$, Lemma \ref{lem:lowerbounderror} says that the lowest error a player could experience is $\var$. This means that the ratio between the two errors is lower than 9 so long as $\ndraw_j \geq \frac{\mue}{9 \cd \var}$. Therefore, in order to get a factor of 9 bound for the overall Price of Anarchy, we need to handle the case of players with size $\leq \frac{\mue}{9\cd \var}$, when players have very few samples.

\begin{restatable}{lemma}{lowerbounderror}
\label{lem:lowerbounderror}
Consider a player $\ndraw_j$ and any set of players $\col$. Then, we can lower bound the error player $j$ recieves by federating with $\col$: 
\begin{equation*}
err_{j}(\col\cup\{\ndraw_j\})\geq \begin{cases}
 \frac{1}{2} \cd\frac{\mue}{\ndraw_j} & \ndraw_j \geq \frac{\mue +\var}{2\var}\\
 \var & \text{otherwise}
\end{cases}
\end{equation*}
\end{restatable}

Lemma \ref{lem:twcspec} is the first of two lemmas handling the case of players with very few samples. It shows that, if a player of size $\leq \frac{\mue}{3\cd \var}$ is federating with a set of players of total size at least $\frac{\mue}{3\cd \var}$, it is possible to \emph{upper} bound on the error of players in $\partition_M$ by $7.5\cd \var$. Given the lower bound of $\var$ in Lemma \ref{lem:lowerbounderror}, these together show that there is a ratio of 7.5 at most between the error this player experiences in its best and worst-case arrangements. 

\begin{restatable}{lemma}{twcspec}
\label{lem:twcspec}
Consider a player $j$ federating with a coalition $\col$. If the total number of samples $\total_{\col}$ is at least $\frac{\mue}{3\var}$, then $err_j(\col \cup \{\ndraw_j\}) \leq 7.25 \cd \var$. 
\end{restatable}

However, Lemma \ref{lem:twcspec} does not handle one situation: what if a player of size $\leq \frac{\mue}{9\cd \var}$ is federating with a group of players of total size $\leq \frac{\mue}{3\cd \var}$? Lemma \ref{lem:relaxed} addresses this last case: it shows that the only such arrangement that is stable is one where all such players are grouped together into a single arrangement. Note that this lemma is itself should not be obvious: it is composed of multiple sub-lemmas which are stated and proved in the appendix. The fact that there can be only one group of such players is used in the Theorem \ref{thrm:PoA} to create an overall bound of $9$. 
\begin{restatable}{lemma}{relaxed}
\label{lem:relaxed}
Consider an arrangement of players, all of size $\leq \frac{\mue}{3\var}$, where at least one player is in a federating cluster where the total mass of its partners is no more than $\frac{\mue}{3\var}$. Then, the only stable arrangement of these players is to have all of them federating together. 
\end{restatable}

The full proof of Theorem \ref{thrm:PoA} uses these lemmas collectively in order to get an overall Price of Anarchy bound of 9, showing that the worst individually stable arrangement has total cost no more than 9 times the optimal cost. 

\section{Conclusion}

In this work, we have given the first Price of Anarchy bound for a game-theoretic model of federated learning.  This bound quantifies a key tension between individual incentives and overall societal goals, answering a key question left open in prior literature. Beyond this bound, we also provide an efficient algorithm to calculate an optimal partition of players into federating coalitions, and have characterized conditions where the Price of Anarchy and/or Price of Stability is equal to 1. 

There are multiple fascinating extensions to this work. To begin with, other definitions of societal cost (for example, weighting players' errors differently) could produce different Price of Anarchy bounds. Additionally, further work could model more sophisticated methods of federation, including models of domain adaptation. Finally, it would be interesting to explore other notions of societal interest. For example, one vein of research is fairness: how are error rates divided among federating players? Questions might revolve around the maximum gap in error rates between players and whether players that contribute more samples are always rewarded with lower error. Beyond these avenues, though, we believe that the broad topic of federated learning will continue to contain multiple useful and interesting research directions. 

\section{Ethics and societal impact}\label{sec:ethics}
Given this work's focus defining notions of optimality, there are important ethical considerations. In particular, \enquote{optimality} can be defined in multiple different ways: Section \ref{sec:opt} motivates the definition we use and Appendix \ref{app:otherdef} discusses the merits of other definitions. In particular, it is worth emphasizing that \enquote{optimality} is a technical term in optimization and game theory which is always with respect to a given objective function and does not imply a more holistic notion of how desirable a certain solution is. For example, an arrangement could be \enquote{optimal} and still be unfair in how errors are distributed among players. 


Although our methodology is application-agnostic, federated learning is a machine learning tool that could be applied towards positive goals (e.g. predicting patient outcomes at hospitals) or negative goals (e.g. used to surveil and control populations). It is also worth considering, for each application, whether there could be some other approach that would better address the need. For example, it may be worth considering whether approaches aiming at increasing the number of samples available for low-resource agents would do a better job of increasing the benefit of a federated learning solution. It may even be the case that a solution beyond machine learning would be preferable, such as interventions to reduce the need for a predictive model. 
\bibliographystyle{plainnat}
\bibliography{biblio.bib}

\begin{thebibliography}{36}
\providecommand{\natexlab}[1]{#1}
\providecommand{\url}[1]{\texttt{#1}}
\expandafter\ifx\csname urlstyle\endcsname\relax
  \providecommand{\doi}[1]{doi: #1}\else
  \providecommand{\doi}{doi: \begingroup \urlstyle{rm}\Url}\fi

\bibitem[Anglano et~al.(2018)Anglano, Canonico, Castagno, Guazzone, and
  Sereno]{Anglano2018AGA}
C.~Anglano, M.~Canonico, Paolo Castagno, Marco Guazzone, and M.~Sereno.
\newblock A game-theoretic approach to coalition formation in fog provider
  federations.
\newblock \emph{2018 Third International Conference on Fog and Mobile Edge
  Computing (FMEC)}, pages 123--130, 2018.

\bibitem[Anshelevich et~al.(2008)Anshelevich, Dasgupta, Kleinberg, Tardos,
  Wexler, and Roughgarden]{anshelevich2008price}
Elliot Anshelevich, Anirban Dasgupta, Jon Kleinberg, {\'E}va Tardos, Tom
  Wexler, and Tim Roughgarden.
\newblock The price of stability for network design with fair cost allocation.
\newblock \emph{SIAM Journal on Computing}, 38\penalty0 (4):\penalty0
  1602--1623, 2008.

\bibitem[Bagdasaryan and Shmatikov(2019)]{bagdasaryan2019differential}
Eugene Bagdasaryan and Vitaly Shmatikov.
\newblock Differential privacy has disparate impact on model accuracy, 2019.

\bibitem[Blum et~al.(2017)Blum, Haghtalab, Procaccia, and Qiao]{collabPAC}
Avrim Blum, Nika Haghtalab, Ariel~D Procaccia, and Mingda Qiao.
\newblock Collaborative pac learning.
\newblock In I.~Guyon, U.~V. Luxburg, S.~Bengio, H.~Wallach, R.~Fergus,
  S.~Vishwanathan, and R.~Garnett, editors, \emph{Advances in Neural
  Information Processing Systems 30}, pages 2392--2401. Curran Associates,
  Inc., 2017.
\newblock URL
  \url{http://papers.nips.cc/paper/6833-collaborative-pac-learning.pdf}.

\bibitem[Blum et~al.(2021)Blum, Haghtalab, Phillips, and Shao]{blum2021one}
Avrim Blum, Nika Haghtalab, Richard~Lanas Phillips, and Han Shao.
\newblock One for one, or all for all: Equilibria and optimality of
  collaboration in federated learning.
\newblock \emph{arXiv preprint arXiv:2103.03228}, 2021.

\bibitem[Bogomolnaia and Jackson(2002)]{BOGOMOLNAIA2002201}
Anna Bogomolnaia and Matthew~O. Jackson.
\newblock The stability of hedonic coalition structures.
\newblock \emph{Games and Economic Behavior}, 38\penalty0 (2):\penalty0 201 --
  230, 2002.
\newblock ISSN 0899-8256.
\newblock \doi{https://doi.org/10.1006/game.2001.0877}.
\newblock URL
  \url{http://www.sciencedirect.com/science/article/pii/S0899825601908772}.

\bibitem[Caldarola et~al.(2021)Caldarola, Mancini, Galasso, Ciccone,
  Rodol{\`a}, and Caputo]{caldarola2021cluster}
Debora Caldarola, Massimiliano Mancini, Fabio Galasso, Marco Ciccone, Emanuele
  Rodol{\`a}, and Barbara Caputo.
\newblock Cluster-driven graph federated learning over multiple domains.
\newblock \emph{arXiv preprint arXiv:2104.14628}, 2021.

\bibitem[Chen et~al.(2021)Chen, Shen, Zhang, and Tang]{chen2021dynamic}
Shengbo Chen, Cong Shen, Lanxue Zhang, and Yuanmin Tang.
\newblock Dynamic aggregation for heterogeneous quantization in federated
  learning.
\newblock \emph{IEEE Transactions on Wireless Communications}, 2021.

\bibitem[Donahue and Kleinberg(2021)]{donahue2020model}
Kate Donahue and Jon Kleinberg.
\newblock Model-sharing games: Analyzing federated learning under voluntary
  participation.
\newblock \emph{AAAI 2021}, 2021.
\newblock URL \url{https://arxiv.org/abs/2010.00753}.

\bibitem[{Duan} et~al.(2021){Duan}, {Liu}, {Chen}, {Liu}, {Tan}, and
  {Liang}]{Duan_selfbalance}
M.~{Duan}, D.~{Liu}, X.~{Chen}, R.~{Liu}, Y.~{Tan}, and L.~{Liang}.
\newblock Self-balancing federated learning with global imbalanced data in
  mobile systems.
\newblock \emph{IEEE Transactions on Parallel and Distributed Systems},
  32\penalty0 (1):\penalty0 59--71, 2021.

\bibitem[Duan et~al.()Duan, Liu, Ji, Liu, Liang, Chen, and Tan]{duanfedgroup}
Moming Duan, Duo Liu, Xinyuan Ji, Renping Liu, Liang Liang, Xianzhang Chen, and
  Yujuan Tan.
\newblock Fedgroup: Efficient clustered federated learning via decomposed
  data-driven measure.

\bibitem[Guazzone et~al.(2014)Guazzone, Anglano, and Sereno]{Guazzone_2014}
Marco Guazzone, Cosimo Anglano, and Matteo Sereno.
\newblock A game-theoretic approach to coalition formation in green cloud
  federations.
\newblock \emph{2014 14th IEEE/ACM International Symposium on Cluster, Cloud
  and Grid Computing}, May 2014.
\newblock \doi{10.1109/ccgrid.2014.37}.
\newblock URL \url{http://dx.doi.org/10.1109/CCGrid.2014.37}.

\bibitem[Guo et~al.(2021)Guo, Wang, Zhou, Jiang, and Patel]{guo2021multi}
Pengfei Guo, Puyang Wang, Jinyuan Zhou, Shanshan Jiang, and Vishal~M Patel.
\newblock Multi-institutional collaborations for improving deep learning-based
  magnetic resonance image reconstruction using federated learning.
\newblock \emph{arXiv preprint arXiv:2103.02148}, 2021.

\bibitem[Hasan(2021)]{hasan2021incentive}
Cengis Hasan.
\newblock Incentive mechanism design for federated learning: Hedonic game
  approach, 2021.

\bibitem[Jamali-Rad et~al.(2021)Jamali-Rad, Abdizadeh, and
  Szabo]{jamali2021federated}
Hadi Jamali-Rad, Mohammad Abdizadeh, and Attila Szabo.
\newblock Federated learning with taskonomy for non-iid data.
\newblock \emph{arXiv preprint arXiv:2103.15947}, 2021.

\bibitem[Kairouz et~al.(2019)Kairouz, McMahan, Avent, Bellet, Bennis, Bhagoji,
  Bonawitz, Charles, Cormode, Cummings, D'Oliveira, Rouayheb, Evans, Gardner,
  Garrett, Gascón, Ghazi, Gibbons, Gruteser, Harchaoui, He, He, Huo,
  Hutchinson, Hsu, Jaggi, Javidi, Joshi, Khodak, Konečný, Korolova,
  Koushanfar, Koyejo, Lepoint, Liu, Mittal, Mohri, Nock, Özgür, Pagh,
  Raykova, Qi, Ramage, Raskar, Song, Song, Stich, Sun, Suresh, Tramèr,
  Vepakomma, Wang, Xiong, Xu, Yang, Yu, Yu, and Zhao]{kairouz2019advances}
Peter Kairouz, H.~Brendan McMahan, Brendan Avent, Aurélien Bellet, Mehdi
  Bennis, Arjun~Nitin Bhagoji, Keith Bonawitz, Zachary Charles, Graham Cormode,
  Rachel Cummings, Rafael G.~L. D'Oliveira, Salim~El Rouayheb, David Evans,
  Josh Gardner, Zachary Garrett, Adrià Gascón, Badih Ghazi, Phillip~B.
  Gibbons, Marco Gruteser, Zaid Harchaoui, Chaoyang He, Lie He, Zhouyuan Huo,
  Ben Hutchinson, Justin Hsu, Martin Jaggi, Tara Javidi, Gauri Joshi, Mikhail
  Khodak, Jakub Konečný, Aleksandra Korolova, Farinaz Koushanfar, Sanmi
  Koyejo, Tancrède Lepoint, Yang Liu, Prateek Mittal, Mehryar Mohri, Richard
  Nock, Ayfer Özgür, Rasmus Pagh, Mariana Raykova, Hang Qi, Daniel Ramage,
  Ramesh Raskar, Dawn Song, Weikang Song, Sebastian~U. Stich, Ziteng Sun,
  Ananda~Theertha Suresh, Florian Tramèr, Praneeth Vepakomma, Jianyu Wang,
  Li~Xiong, Zheng Xu, Qiang Yang, Felix~X. Yu, Han Yu, and Sen Zhao.
\newblock Advances and open problems in federated learning, 2019.

\bibitem[Koutsoupias and Papadimitriou(1999)]{koutsoupias1999worst}
Elias Koutsoupias and Christos Papadimitriou.
\newblock Worst-case equilibria.
\newblock In \emph{Annual Symposium on Theoretical Aspects of Computer
  Science}, pages 404--413. Springer, 1999.

\bibitem[Kumar et~al.(2021)Kumar, Khan, Kumar, Zakria, Golilarz, Zhang, Ting,
  Zheng, and Wang]{kumar2021blockchain}
Rajesh Kumar, Abdullah~Aman Khan, Jay Kumar, A~Zakria, Noorbakhsh~Amiri
  Golilarz, Simin Zhang, Yang Ting, Chengyu Zheng, and WenYong Wang.
\newblock Blockchain-federated-learning and deep learning models for covid-19
  detection using ct imaging.
\newblock \emph{IEEE Sensors Journal}, 2021.

\bibitem[Laguel et~al.(2021)Laguel, Pillutla, Malick, and
  Harchaoui]{laguel2021superquantile}
Yassine Laguel, Krishna Pillutla, Jer{\^o}me Malick, and Zaid Harchaoui.
\newblock A superquantile approach to federated learning with heterogeneous
  devices.
\newblock In \emph{2021 55th Annual Conference on Information Sciences and
  Systems (CISS)}, pages 1--6. IEEE, 2021.

\bibitem[Le et~al.(2021)Le, Tran, Tun, Nguyen, Pandey, Han, and
  Hong]{incentivemech}
Tra Huong~Thi Le, Nguyen~H. Tran, Yan~Kyaw Tun, Minh N.~H. Nguyen, Shashi~Raj
  Pandey, Zhu Han, and Choong~Seon Hong.
\newblock An incentive mechanism for federated learning in wireless cellular
  network: An auction approach.
\newblock \emph{IEEE Transactions on Wireless Communications}, pages 1--1,
  2021.
\newblock \doi{10.1109/TWC.2021.3062708}.

\bibitem[Lee et~al.(2020)Lee, Oh, Shin, and Yoon]{Lee2020AccurateAF}
J.~Lee, Jaehoon Oh, Yooju Shin, and S.~Yoon.
\newblock Accurate and fast federated learning via iid and communication-aware
  grouping.
\newblock \emph{ArXiv}, abs/2012.04857, 2020.

\bibitem[Li et~al.(2019)Li, Sanjabi, Beirami, and Smith]{li2019fair}
Tian Li, Maziar Sanjabi, Ahmad Beirami, and Virginia Smith.
\newblock Fair resource allocation in federated learning, 2019.

\bibitem[Li et~al.(2020)Li, Sahu, Talwalkar, and Smith]{Li_2020}
Tian Li, Anit~Kumar Sahu, Ameet Talwalkar, and Virginia Smith.
\newblock Federated learning: Challenges, methods, and future directions.
\newblock \emph{IEEE Signal Processing Magazine}, 37\penalty0 (3):\penalty0
  50–60, May 2020.
\newblock ISSN 1558-0792.
\newblock \doi{10.1109/msp.2020.2975749}.
\newblock URL \url{http://dx.doi.org/10.1109/MSP.2020.2975749}.

\bibitem[{Lim} et~al.(2020){Lim}, {Luong}, {Hoang}, {Jiao}, {Liang}, {Yang},
  {Niyato}, and {Miao}]{fedsurvey}
W.~Y.~B. {Lim}, N.~C. {Luong}, D.~T. {Hoang}, Y.~{Jiao}, Y.~C. {Liang},
  Q.~{Yang}, D.~{Niyato}, and C.~{Miao}.
\newblock Federated learning in mobile edge networks: A comprehensive survey.
\newblock \emph{IEEE Communications Surveys Tutorials}, 22\penalty0
  (3):\penalty0 2031--2063, 2020.
\newblock \doi{10.1109/COMST.2020.2986024}.

\bibitem[Lin et~al.(2018)Lin, Stich, Patel, and Jaggi]{lin2018dont}
Tao Lin, Sebastian~U. Stich, Kumar~Kshitij Patel, and Martin Jaggi.
\newblock Don't use large mini-batches, use local sgd, 2018.

\bibitem[Liu et~al.(2020)Liu, Zhang, Song, and Letaief]{Liu_2020}
Lumin Liu, Jun Zhang, S.H. Song, and Khaled~B. Letaief.
\newblock Client-edge-cloud hierarchical federated learning.
\newblock \emph{ICC 2020 - 2020 IEEE International Conference on Communications
  (ICC)}, Jun 2020.
\newblock \doi{10.1109/icc40277.2020.9148862}.
\newblock URL \url{http://dx.doi.org/10.1109/icc40277.2020.9148862}.

\bibitem[McMahan et~al.(2016)McMahan, Moore, Ramage, Hampson, and
  y~Arcas]{mcmahan2016communicationefficient}
H.~Brendan McMahan, Eider Moore, Daniel Ramage, Seth Hampson, and
  Blaise~Agüera y~Arcas.
\newblock Communication-efficient learning of deep networks from decentralized
  data, 2016.

\bibitem[Mohri et~al.(2019)Mohri, Sivek, and Suresh]{mohri2019agnostic}
Mehryar Mohri, Gary Sivek, and Ananda~Theertha Suresh.
\newblock Agnostic federated learning, 2019.

\bibitem[Papadimitriou(2001)]{papadimitriou2001algorithms}
Christos Papadimitriou.
\newblock Algorithms, games, and the internet.
\newblock In \emph{Proceedings of the thirty-third annual ACM symposium on
  Theory of computing}, pages 749--753, 2001.

\bibitem[Sattler et~al.(2020)Sattler, Muller, and Samek]{Sattler_2020}
Felix Sattler, Klaus-Robert Muller, and Wojciech Samek.
\newblock Clustered federated learning: Model-agnostic distributed multitask
  optimization under privacy constraints.
\newblock \emph{IEEE Transactions on Neural Networks and Learning Systems},
  page 1–13, 2020.
\newblock ISSN 2162-2388.
\newblock \doi{10.1109/tnnls.2020.3015958}.
\newblock URL \url{http://dx.doi.org/10.1109/TNNLS.2020.3015958}.

\bibitem[{Shlezinger} et~al.(2020){Shlezinger}, {Rini}, and
  {Eldar}]{ShlezingerClustFed}
N.~{Shlezinger}, S.~{Rini}, and Y.~C. {Eldar}.
\newblock The communication-aware clustered federated learning problem.
\newblock In \emph{2020 IEEE International Symposium on Information Theory
  (ISIT)}, pages 2610--2615, 2020.

\bibitem[Vaid et~al.(2021)Vaid, Jaladanki, Xu, Teng, Kumar, Lee, Somani,
  Paranjpe, De~Freitas, Wanyan, et~al.]{vaid2021federated}
Akhil Vaid, Suraj~K Jaladanki, Jie Xu, Shelly Teng, Arvind Kumar, Samuel Lee,
  Sulaiman Somani, Ishan Paranjpe, Jessica~K De~Freitas, Tingyi Wanyan, et~al.
\newblock Federated learning of electronic health records to improve mortality
  prediction in hospitalized patients with covid-19: Machine learning approach.
\newblock \emph{JMIR medical informatics}, 9\penalty0 (1):\penalty0 e24207,
  2021.

\bibitem[Wang et~al.(2020)Wang, Hsu, Diaz, and Calmon]{wang2020split}
Hao Wang, Hsiang Hsu, Mario Diaz, and Flavio~P. Calmon.
\newblock To split or not to split: The impact of disparate treatment in
  classification, 2020.

\bibitem[Xia et~al.(2021)Xia, Yang, Li, Myronenko, Xu, Obinata, Mori, An,
  Harmon, Turkbey, et~al.]{xia2021auto}
Yingda Xia, Dong Yang, Wenqi Li, Andriy Myronenko, Daguang Xu, Hirofumi
  Obinata, Hitoshi Mori, Peng An, Stephanie Harmon, Evrim Turkbey, et~al.
\newblock Auto-fedavg: Learnable federated averaging for multi-institutional
  medical image segmentation.
\newblock \emph{arXiv preprint arXiv:2104.10195}, 2021.

\bibitem[Yu et~al.(2020)Yu, Bagdasaryan, and Shmatikov]{yu2020salvaging}
Tao Yu, Eugene Bagdasaryan, and Vitaly Shmatikov.
\newblock Salvaging federated learning by local adaptation, 2020.

\bibitem[Zhang et~al.(2021)Zhang, Zhou, Lu, Wang, Zhu, Sun, Wang, Lo, and
  Wang]{zhang2021dynamic}
Weishan Zhang, Tao Zhou, Qinghua Lu, Xiao Wang, Chunsheng Zhu, Haoyun Sun,
  Zhipeng Wang, Sin~Kit Lo, and Fei-Yue Wang.
\newblock Dynamic fusion-based federated learning for covid-19 detection.
\newblock \emph{IEEE Internet of Things Journal}, 2021.

\end{thebibliography}


\ifthenelse{\equal{\preprint}{1}}{
\section*{Checklist}

\begin{enumerate}

\item For all authors...
\begin{enumerate}
  \item Do the main claims made in the abstract and introduction accurately reflect the paper's contributions and scope?
    \answerYes{The main contributions of our work is analysis of optimality (Section \ref{sec:opt}) and Price of Anarchy (Section \ref{sec:PoA}) in a theoretical model of federated learning}. 
  \item Did you describe the limitations of your work?
    \answerYes{Section \ref{sec:model} and Appendix \ref{app:otherdef} describe assumptions and limitations of our model and our work}. 
  \item Did you discuss any potential negative societal impacts of your work?
    \answerYes{Section \ref{sec:ethics} describes some potential negative implications and ethical concerns of our work. }
  \item Have you read the ethics review guidelines and ensured that your paper conforms to them?
    \answerYes{}
\end{enumerate}

\item If you are including theoretical results...
\begin{enumerate}
  \item Did you state the full set of assumptions of all theoretical results?
    \answerYes{We state assumptions in Section \ref{sec:model}.}
	\item Did you include complete proofs of all theoretical results?
    \answerYes{All of the proofs are given in Appendices \ref{app:opt} and \ref{app:poa}.}
\end{enumerate}

\item If you ran experiments...
\begin{enumerate}
  \item Did you include the code, data, and instructions needed to reproduce the main experimental results (either in the supplemental material or as a URL)?
    \answerNA{}
  \item Did you specify all the training details (e.g., data splits, hyperparameters, how they were chosen)?
    \answerNA{}
	\item Did you report error bars (e.g., with respect to the random seed after running experiments multiple times)?
    \answerNA{}
	\item Did you include the total amount of compute and the type of resources used (e.g., type of GPUs, internal cluster, or cloud provider)?
    \answerNA{}
\end{enumerate}

\item If you are using existing assets (e.g., code, data, models) or curating/releasing new assets...
\begin{enumerate}
  \item If your work uses existing assets, did you cite the creators?
    \answerYes{We use a model from \cite{donahue2020model}, who we cite. We describe their contribution in Section \ref{sec:model}.}
  \item Did you mention the license of the assets?
    \answerNA{}
  \item Did you include any new assets either in the supplemental material or as a URL?
    \answerNA{}
  \item Did you discuss whether and how consent was obtained from people whose data you're using/curating?
    \answerNA{}
  \item Did you discuss whether the data you are using/curating contains personally identifiable information or offensive content?
    \answerNA{}
\end{enumerate}

\item If you used crowdsourcing or conducted research with human subjects...
\begin{enumerate}
  \item Did you include the full text of instructions given to participants and screenshots, if applicable?
    \answerNA{}
  \item Did you describe any potential participant risks, with links to Institutional Review Board (IRB) approvals, if applicable?
    \answerNA{}
  \item Did you include the estimated hourly wage paid to participants and the total amount spent on participant compensation?
    \answerNA{}
\end{enumerate}

\end{enumerate}}
{}

\newpage 
\appendix

\section{Alternate definitions of optimality}\label{app:otherdef}

The definition of cost used in this paper is given in Definition \ref{def:opt}, which says that an arrangement is optimal if it minimizes the weighted sum of errors over players. As discussed previously, this definition is well-motivated by existing federated learning literature. Additionally, it matches the societal good perspective when the unit society cares about is at the level of the data point. For example, consider the example when the federating agents are hospitals and data points represent individual patients. Then, society as a whole likely cares about minimizing the overall error patients experience, which corresponds to the per-data-point notion of error. 

However, other cost functions are worth discussing. For example, Definition \ref{def:unweight} gives an unweighted notion of error: 

\begin{definition}[Unweighted cost]\label{def:unweight}
The unweighted cost function is given by summing the error over each of the players, without any weighting with respect to size: 
$$f_{u}(\Pi) = \sum_{\col \in \partition}f_{u}(\col) = \sum_{\col \in \partition}\sum_{i \in \col} err_{i}(\col)$$
\end{definition}

This definition might be better in a model where the unit society cares about is at the level of the agent. For example, consider a situation where the individual federating agent is a cell phone owned by a single person and data points are word predictions. Then, society as a whole might care about minimizing the sum of errors that individual cell phone users experience, which is given by the unweighted error function. 

Finally, we may wish to consider some completely different weight function, given by the definition below: 

\begin{definition}[Arbitrary weights]\label{def:arbweight}
The arbitrary cost metric is given by summing the weight over each of the players according to some weight $\sum_{i \in [\nplayer]} p_i = 1$ 
$$f_{a}(\Pi) = \sum_{\col \in \partition}f_{a}(\col) = \sum_{\col \in \partition}\sum_{i \in \col} p_i \cd err_{i}(\col)$$
\end{definition}

Definitions like this have been analyzed in \cite{li2019fair, mohri2019agnostic, laguel2021superquantile, chen2021dynamic}. For example, the set of weights $\{p_i\}$ could have fairness goals, attempting to up-weight players with higher error. Alternatively, it could represent some notion of the data quality players are contributing, with players producing more or lower-error players being weighted more. 

In this work, we selected Definition \ref{def:opt} (weighted error) based on its standard use in the federated learning literature. Analysis of the same type (calculating an optimal arrangement and analyzing the Price of Anarchy) could be completed for any other definition of cost, but would require new proofs for calculation of optimal arrangements and for any Price of Anarchy bound.

\section{Optimality calculation}\label{app:opt}

\optdef*

\begin{proof}
\begin{align*}
\costw(\col) &= \sum_{j\in \col}err_j(\col) \cd \ndraw_i = \sum_{j\in \col}\p{\frac{\mue}{\sum_{i \in \col}\ndraw_i}+ \var \frac{\sum_{i\ne j}\ndraw_i^2 + \p{\sum_{i\ne j}\ndraw_i}^2}{\p{\sum_{i\in \col}\ndraw_i}^2}}\cd \ndraw_j\\
&= \sum_{j\in \col}\frac{\mue}{\sum_{i\in\col}\ndraw_i} \cd \ndraw_j + \var \sum_{j\in \col} \frac{\sum_{i\ne j}\ndraw_i^2 + \p{\sum_{i\ne j}\ndraw_i}^2}{\p{\sum_{i\in \col}\ndraw_i}^2}\cd \ndraw_j\\
&= \mue  + \var \sum_{j \in \col}\frac{\ndraw_j \cd \sum_{i \ne j}\ndraw_i^2 + \ndraw_j \cd (\total_{\col} - \ndraw_j)^2}{\total_{\col}^2}
\end{align*}
where we have used $\total_{\col} = \sum_{i \in \col}\ndraw_i$. Focusing solely on the numerator of the second term, we simplify:
\begin{align*}
&\sum_{j\in \col}\cb{\ndraw_j \cd \sum_{i\ne j} \ndraw_i^2 + \ndraw_j \cd \total_{\col}^2 + \ndraw_j^3 - 2 \total_{\col} \cd \ndraw_j^2} = \sum_{j \in \col} \ndraw_j \cd \sum_{i\in \col} \ndraw_i^2 + \total_{\col}^2\sum_{j \in \col}\ndraw_j -2 \total_{\col} \sum_{j \in \col} \ndraw_j^2\\
&=\total_{\col} \cd \sum_{i \in \col}\ndraw_i^2 + \total_{\col}^3 - 2 \total_{\col} \sum_{i \in \col}\ndraw_i^2=\total_{\col}^3 - \total_{\col} \cd \sum_{i \in \col}\ndraw_i^2
\end{align*}
Combining this with the rest of the term gives: 
$$\mue + \var \cd \frac{\total_{\col}^3 - \total_{\col} \cd \sum_{i \in \col} \ndraw_i^2}{\total_{\col}^2} =  \mue + \var \cd \total_{\col} - \var \frac{\sum_{i \in \col} \ndraw_i^2}{\total_{\col}}$$
\end{proof}

\alonebad*
\begin{proof}
We will prove this result by the setting where $\nplayer$ players each have $\ndraw$ samples, for $\nplayer > \rho$ and any $\mue, \var, \ndraw \in \mathbb{N}_{\geq 1}$ such that $\ndraw < \p{\frac{\nplayer}{\rho}-1} \frac{\mue}{(\nplayer-1) \cd \var}$.

In this simplified setting where all of the players have the same number of samples, the cost of a coalition $\col$ involving $\nplayer$ players is given by: 
$$\mue + \var \cd \ndraw \cd \nplayer - \var \cd \frac{\nplayer \cd \ndraw^2}{\nplayer \cd \ndraw} = \mue + \var \cd \ndraw \cd (\nplayer -1)$$
For our given example, $\ndraw < \frac{\mue}{\var}$, which implies that \enquote{merging} any two groups $A$ and $B$ will reduce total cost: 
\begin{align*}
\costw(A) + \costw(B) &= \mue + \var \cd \ndraw \cd (\nplayer_A -1)+ \mue + \var \cd \ndraw \cd (\nplayer_B -1) \\
&> \mue + \var \cd \ndraw \cd (\nplayer_A + \nplayer_B -1) \\
&= \costw(A \cup B)
\end{align*}
This implies that the optimal cost is achieved by $\gcol$, given by $\mue + \var \cd (\nplayer -1)$. Conversely, the cost of having $\nplayer$ players doing local learning is:
$$\costw(\alone) = \sum_{i=1}^{\nplayer}\cb{\mue + \var \cd \ndraw - \var \cd \frac{\ndraw^2}{\ndraw}} = \sum_{i=1}^{\nplayer} \mue = \mue \cd \nplayer$$
Combining these facts gives:
\begin{align*}
\frac{\costw(\alone)}{\costw(OPT)} &=\frac{\mue \cd \nplayer}{\mue + \var \cd (\nplayer-1) \cd \ndraw} = \frac{\nplayer}{1 + \frac{\var}{\mue} \cd (\nplayer-1) \cd \ndraw}\\  
 &> \frac{\nplayer}{1 + \frac{\var}{\mue} \cd (\nplayer-1) \cd \frac{\mue}{(\nplayer-1) \cd \var }\cd \p{\frac{\nplayer}{\rho} -1}}\\
 &= \frac{\nplayer}{1 + \frac{\nplayer}{\rho} -1} = \rho
\end{align*}
as desired. 
\end{proof}

\gcolbad*
\begin{proof}

We will prove this result by the setting where $\nplayer$ players each have $\ndraw$ samples, with $\nplayer > \frac{1}{\rho}$ and any $\mue, \var, \ndraw \in \mathbb{N}_{\geq 1}$ such that $\ndraw >\max \br{ \frac{\mue}{\var \cd (\nplayer-1)} \cd \p{\rho \cd \nplayer - 1}, \frac{\mue}{\var}}$. 

The initial construction follows similarly to Lemma \ref{lem:alonebad}. For our given example, $\ndraw > \frac{\mue}{\var}$, which implies that \enquote{merging} any two groups $A$ and $B$ will \emph{increase} total cost: 
\begin{align*}
\costw(A) + \costw(B) &= \mue + \var \cd \ndraw \cd (\nplayer_A -1)+ \mue + \var \cd \ndraw \cd (\nplayer_B -1)\\
&< \mue + \var \cd \ndraw \cd (\nplayer_A + \nplayer_B -1) \\
&= \costw(A \cup B)
\end{align*}
This implies that the optimal cost is achieved $\alone$. Using the value derived in the proof of Lemma \ref{lem:alonebad}, we have: 
\begin{align*}
\frac{\costw(\gcol)}{\costw(OPT)} &=\frac{\mue + \var \cd (\nplayer-1) \cd \ndraw}{\mue \cd \nplayer} \\
&= \frac{1 + \frac{\var}{\mue} \cd (\nplayer -1) \cd \ndraw}{\nplayer} \\
&> \frac{1 + \frac{\var}{\mue} \cd (\nplayer -1) \cd \max \br{ \frac{\mue}{\var \cd (\nplayer-1)} \cd \p{\rho \cd \nplayer - 1}, \frac{\mue}{\var}}}{\nplayer} \\
&\geq \frac{1 + \frac{\var}{\mue} \cd (\nplayer -1) \cd \frac{\mue}{\var \cd (\nplayer-1)} \cd \p{\rho \cd \nplayer - 1}}{\nplayer}\\
&=\frac{1 +\rho \cd \nplayer - 1}{\nplayer} = \rho
\end{align*}
as desired. 
\end{proof}

The proof of Theorem \ref{thrm:optcalc}, below, relies on multiple sub-lemmas which are stated and proved immediately afterwards. 

\optcalc*

\begin{proof}
First, we note two special cases. If $\{\ndraw_i\}\leq \frac{\mue}{\var}$, then by Lemma \ref{lem:gcolcore} (stated and proved later in this appendix) the grand coalition $\gcol$ is core stable. For the grand coalition, core stability implies individual stability, so we know that every player prefers $\gcol$ to local learning. This implies that, following the steps given in this theorem, every player will prefer to join the growing coalition as opposed to doing local learning, and so the optimal arrangement is $\gcol$. 

Next, if $\{\ndraw_i\}> \frac{\mue}{\var}$, then by Lemma 5.3 in \cite{donahue2020model} every player minimizes their error in $\alone$ (local learning). As a result, using the algorithm given in the statement of this theorem, every player will increase their error by combining with another player, so $\alone$ is optimal. If $\{\ndraw_i\}\geq \frac{\mue}{\var}$ (some players have exactly $\frac{\mue}{\var}$ samples), then all players with $\ndraw_i = \frac{\mue}{\var}$ will be indifferent towards being merged with any other player also of size $\frac{\mue}{\var}$, but no player of size strictly greater than $\frac{\mue}{\var}$ will be able to be merged. The resulting optimal arrangement will have all of the players of size exactly  $\frac{\mue}{\var}$ together, with all other players doing local learning, and will have cost identical to $\alone$. 

Finally, we will consider the case where some players have size strictly less than $\frac{\mue}{\var}$ and some have strictly more. Call the partition calculated by following the steps of this theorem $\partition$, and consider any other coalition partition $\partition'$. We will convert $\partition'$ into $\partition$ using only cost reducing or maintaining steps, which will show that $\partition$ is optimal. We will refer to players with size $\leq \frac{\mue}{\var}$ as small, and players of size $> \frac{\mue}{\var}$ as large. 
\begin{itemize}
    \item If there are any coalitions where players would prefer to leave the coalition, remove them in order of descending size. Note: a coalition made up of only players of size smaller than $\frac{\mue}{\var}$ will never have players leave. A coalition made up of only players of size larger than $\frac{\mue}{\var}$ will always wish to have players leave. This reduces total cost by Lemma \ref{lem:addminsame}. 
    \item Every coalition of size 2 or larger will have at least one small player in it. Begin merging all such coalitions (as well as any small players doing local learning), removing large players as necessary (in descending size, if they would prefer local learning). Note that the merging operation will never remove a small player, so it always strictly reduces the number of coalitions involving small players. This reduces cost by Lemma \ref{lem:merge}. 
    \item When all of the small players are in one coalition, if there are large players in the coalition as well, check if they are the smallest possible large player. If not, swap them for smaller large players iteratively (ones that are doing local learning) until the players in the coalition are doing local learning. By Lemma \ref{lem:swap}, this reduces cost.  
    \item Add large players in increasing order of size (if any wish to join). From Lemma \ref{lem:orderjoin} we know that if player $\ndraw_i$ doesn't wish to join a coalition, then neither will any player of size $\ndraw_j\geq \ndraw_i$. From Lemma \ref{lem:addminsame}, adding any player that wishes to join reduces total cost. 
    \item If no players wish to join, then remove large players in descending order of size if they would prefer local learning, which again from Lemma \ref{lem:addminsame} would reduce cost. From Lemma \ref{lem:wontleave}, if a player of size $\ndraw_i$ doesn't wish to leave, then all other players of size $\ndraw_j\leq \ndraw_i$ also do not wish to leave.
\end{itemize}
The final arrangement exactly matches $\partition$. 
\end{proof}

\addminsame*
\begin{proof}
This proof will work by showing the forms of the inequalities are identical. We will start with the cost inequality:
\begin{align*}
\costw(\{\ndraw_j\}) + \costw(Q) &\geq \costw(\{\ndraw_j\} \cup Q)  \\
\mue + \mue + \var \cd \total_Q - \var \cd \frac{\sum_{i \in Q}\ndraw_i^2}{\total_Q}&\geq\mue + \var \cd \total_Q + \var \cd \ndraw_j - \var \frac{\sum_{i\in Q}\ndraw_i^2 + \ndraw_j^2}{\total_Q + \ndraw_j}\\
\mue &\geq \var \cd \ndraw_j - \var \frac{\sum_{i\in Q}\ndraw_i^2 + \ndraw_j^2}{\total_Q + \ndraw_j}+ \var \cd \frac{\sum_{i \in Q}\ndraw_i^2}{\total_Q}
\end{align*}
Bringing all terms over common denominator on the righthand side: 
\begin{align*}
\mue &\geq  \var \frac{\ndraw_j \cd (\total_Q + \ndraw_j) \cd \total_Q - \total_Q \sum_{i \in Q}\ndraw_i^2 - \ndraw_j^2 \cd \total_Q + \total_Q \sum_{i \in Q} \ndraw_i^2 + \ndraw_j \cd \sum_{i \in Q}\ndraw_i^2}{(\total_Q + \ndraw_j) \cd \total_Q}\\   
\mue &\geq  \var \frac{\ndraw_j \cd (\total_Q + \ndraw_j) \cd \total_Q - \ndraw_j^2 \cd \total_Q + \ndraw_j \cd \sum_{i \in Q}\ndraw_i^2}{(\total_Q + \ndraw_j) \cd \total_Q}\\
\mue &\geq  \var \frac{\ndraw_j \cd \total_Q^2 + \ndraw_j^2 \cd \total_Q - \ndraw_j^2 \cd \total_Q + \ndraw_j \cd \sum_{i \in Q}\ndraw_i^2}{(\total_Q + \ndraw_j) \cd \total_Q}\\
\mue &\geq  \var \frac{\ndraw_j \cd \total_Q^2 + \ndraw_j \cd \sum_{i \in Q}\ndraw_i^2}{(\total_Q + \ndraw_j) \cd \total_Q}\\
\mue &\geq  \var \cd \frac{\ndraw_j}{\total_Q + \ndraw_j} \cd \frac{\total_Q^2 + \sum_{i \in Q}\ndraw_i^2}{\total_Q}
\end{align*}
Next, we will reduce the error inequality to the same form: 
\begin{align*}
err_j(\{\ndraw_j\}) &\geq err_j(\{\ndraw_j\} \cup Q) \\
\frac{\mue}{\ndraw_j} &\geq  \frac{\mue}{\total_Q + \ndraw_j} + \var \cd \frac{\sum_{i \in Q}\ndraw_i^2 + \total_Q^2}{(\total_Q + \ndraw_j)^2}\\
\frac{\mue}{\ndraw_j} -\frac{\mue}{\total_Q + \ndraw_j}&\geq  \var \cd \frac{\sum_{i \in Q}\ndraw_i^2 + \total_Q^2}{(\total_Q + \ndraw_j)^2}\\
\mue \cd \frac{\total_Q}{\ndraw_j \cd (\total_Q +\ndraw_j)}&\geq  \var \cd \frac{\sum_{i \in Q}\ndraw_i^2 + \total_Q^2}{(\total_Q + \ndraw_j)^2}\\
\mue&\geq  \var \cd \frac{(\total_Q + \ndraw_j) \cd \ndraw_j}{\total_Q} \frac{\sum_{i \in Q}\ndraw_i^2 + \total_Q^2}{(\total_Q + \ndraw_j)^2}\\
\mue &\geq  \var \cd \frac{\ndraw_j}{\total_Q + \ndraw_j} \cd \frac{\total_Q^2 + \sum_{i \in Q}\ndraw_i^2}{\total_Q}
\end{align*}
as desired. 
\end{proof}

\swap*

\begin{proof}
We write out each side: 
\begin{align*}
\mue + \var \cd \total_Q + \var \cd \ndraw_j - \var \frac{\sum_{i \in Q} \ndraw_i^2 + \ndraw_j^2}{\total_Q + \ndraw_j} + \mue &> \mue + \var \cd \total_Q + \var \ndraw_k - \var \frac{\sum_{i \in Q}\ndraw_i^2 + \ndraw_k^2}{\total_Q + \ndraw_k} + \mue   \\
 \var \cd \ndraw_j - \var\cd  \frac{\sum_{i \in Q} \ndraw_i^2 + \ndraw_j^2}{\total_Q + \ndraw_j} &> \var \cd \ndraw_k - \var \cd \frac{\sum_{i \in Q}\ndraw_i^2 + \ndraw_k^2}{\total_Q + \ndraw_k}
\end{align*}
Dropping the common $\var$ term for clarity: 
$$ \ndraw_j -  \frac{\sum_{i \in Q} \ndraw_i^2 + \ndraw_j^2}{\total_Q + \ndraw_j} > \ndraw_k - \frac{\sum_{i \in Q}\ndraw_i^2 + \ndraw_k^2}{\total_Q + \ndraw_k}$$
In order to prove the above inequality, we will consider the following fraction:
$$x-  \frac{\sum_{i \in Q} \ndraw_i^2 + x^2}{\total_Q + x}$$
and want to show this is increasing with respect to $x$. The derivative of the function gives:
\begin{align*}
&1 - \frac{2x \cd (\total_Q + x) - (\sum_{i \in Q} \ndraw_i^2 +x^2) \cd 1}{(\total_Q + x)^2} \\
&=\frac{1}{(\total_Q + x)^2}\cd \p{\total_Q^2 + x^2 + 2 x\cd \total_Q - \p{2x \cd (\total_Q + x) - \p{\sum_{i \in Q} \ndraw_i^2 +x^2}}}\\
&=\frac{1}{(\total_Q + x)^2}\cd \p{\total_Q^2 + x^2 + 2 x\cd \total_Q - 2x \cd \total_Q - 2x^2 + \p{\sum_{i \in Q} \ndraw_i^2 +x^2} }\\
&=\frac{1}{(\total_Q + x)^2}\cd \p{\total_Q^2+ \sum_{i \in Q} \ndraw_i^2 }
\end{align*}
which is positive, as desired. This implies that the original inequality is satisfied, meaning that the swapping of the roles of players $j, k$ decreases total cost. 
\end{proof}

\orderjoin*

\begin{proof}
The initial premise depends on whether or not the below inequality is satisfied: 
\begin{align*}
err_j(Q\cup \{\ndraw_j\}) &\geq err_j(\{\ndraw_j\})\\
\frac{\mue}{\total_Q + \ndraw_j} + \var \frac{\sum_{i \in Q} \ndraw_i^2 + \p{\sum_{i \in Q}\ndraw_i}^2}{\p{\total_Q + \ndraw_j}^2}&\geq \frac{\mue}{\ndraw_j}
\end{align*}
Rearranging:
\begin{align*}
\var \frac{\sum_{i \in Q} \ndraw_i^2 + \p{\sum_{i \in Q}\ndraw_i}^2}{\p{\total_Q + \ndraw_j}^2}&\geq \frac{\mue}{\ndraw_j}-\frac{\mue}{\total_Q + \ndraw_j}   \\
 \var \p{\sum_{i \in Q} \ndraw_i^2 + \p{\sum_{i \in Q}\ndraw_i}^2}&\geq \p{\total_Q + \ndraw_j}^2\cd \p{\frac{\mue}{\ndraw_j}-\frac{\mue}{\total_Q + \ndraw_j}}\\
 \var \p{\sum_{i \in Q} \ndraw_i^2 + \p{\sum_{i \in Q}\ndraw_i}^2}&\geq \p{\total_Q + \ndraw_j}^2\cd\frac{\mue \cd \total_Q}{\ndraw_j\cd (\total_Q + \ndraw_j)}\\
 \var \p{\sum_{i \in Q} \ndraw_i^2 + \p{\sum_{i \in Q}\ndraw_i}^2}&\geq \p{\total_Q + \ndraw_j}\cd \frac{\mue \cd \total_Q}{\ndraw_j}\\
 \var \p{\sum_{i \in Q} \ndraw_i^2 + \p{\sum_{i \in Q}\ndraw_i}^2}&\geq \mue \cd \frac{\total_Q^2}{\ndraw_j} + \mue \cd \total_Q
\end{align*}
The lefthand side is a constant independent of $\ndraw_k$ and the righthand side is a constant plus a term that is decreasing in $\ndraw_j$. 
If the original inequality ($err_j(Q\cup \{\ndraw_j\}) \geq err_j(\{\ndraw_j\})$) is satisfied, then it will also be satisfied for any $\ndraw_k \geq \ndraw_j$ (implying $err_k(Q\cup \{\ndraw_k\}) \geq err_k(\{\ndraw_k\})$). Conversely, if the original inequality is not satisfied (so $err_j(Q\cup \{\ndraw_j\}) \leq err_j(\{\ndraw_j\})$), then it will also not be satisfied for any $\ndraw_k \leq \ndraw_j$ (implying $err_k(Q\cup \{\ndraw_k\}) \leq err_k(\{\ndraw_k\})$). 
\end{proof}

\wontleave*

\begin{proof}
First, we will prove the first statement. Suppose by contradiction that some $\ndraw_j$ wishes to leave, but another player of size $\ndraw_k \geq \ndraw_j$ does not wish to. First, we remove $\ndraw_j$ for local learning, which by Lemma \ref{lem:addminsame} reduces total cost. Next, we swap the role of players $j$ and $k$, which by Lemma \ref{lem:swap} again reduces or keeps constant total cost. We have constructed a series of operations that either reduce or keep constant total cost, and results in an arrangement equivalent to simply removing player $j$. By Lemma \ref{lem:addminsame}, this means that player $j$ originally would have wished to leave.

Next, we will prove the second statement. Suppose by contradiction that some player $k$ wishes to leave, even though another player $\ndraw_j > \ndraw_k$ does not wish to leave. First, we remove player $k$ to local learning: if it wishes to leave, then by Lemma \ref{lem:addminsame} removing it reduces or keeps constant total cost. Then, by Lemma \ref{lem:swap}, we can reduce total cost by swapping it with the $\ndraw_j>\ndraw_k$ player. We have constructed a series of operations that either reduce or keep constant total cost, and results in an arrangement equivalent to simply removing player $j$. But this is exactly equivalent to just removing the $\ndraw_j$ player, which we know from Lemma \ref{lem:addminsame} must not reduce total cost (or else player $j$ would wish to leave).
\end{proof}

\merge*

\begin{proof}
First, we have to reason about what $L$ could be. We will say that player $j$ with $\ndraw_j$ samples is a largest element in $P \cup Q$, and WLOG $j \in P$. (If multiple players have $\ndraw_j$ samples, it suffices to select one at random.). We will show that, in order to show Equation \ref{eq:submod0}, it suffices to show that:
\begin{equation}\label{eq:submod1}
\costw(Q) + \costw(P) \geq \costw(Q\cup P\setminus \ndraw_j) + \costw(\{\ndraw_j\})
\end{equation}
First, assume $L$ is empty. Then, every player wishes to stay in the final group. Then, Equation \ref{eq:submod0} becomes:
$$\costw(Q) + \costw(P) > \costw(Q\cup P)$$
From Lemma \ref{lem:addminsame}, we know that because player $\ndraw_j$ doesn't wish to leave $Q\cup P$, removing it must increase total cost:
$$\costw(Q\cup P\setminus \{\ndraw_j\}) + \costw(\{\ndraw_j\})>\costw(Q\cup P)$$
So, if we show that Equation \ref{eq:submod1} is satisfied, then this implies that Equation \ref{eq:submod0} is satisfied. \\
Next, we will assume $L = \{\ndraw_j\}$. Then, the statement we are trying to show is exactly Equation \ref{eq:submod1}.\\
Finally, let's assume that $\vert L\vert \geq 2$: $\ndraw_j$ is removed, but so are some others. Again, by Lemma \ref{lem:addminsame}, because these players prefer local learning to federation, adding them back in to the coalition increases cost, so 
$$\costw(Q) + \costw(P\cup \{\ndraw_j\}) >\costw(Q\cup P\setminus \ndraw_j) + \costw(\{\ndraw_j\})$$
So, it suffices to consider Equation \ref{eq:submod1}: if we prove that this is satisfied, it always implies that Equation \ref{eq:submod0} is satisfied.  

Next, we will prove this statement: 
$$\costw(Q) + \costw(P) \geq \costw(Q\cup P\setminus \ndraw_j) + \costw(\{\ndraw_j\})$$
Plugging in for the form of $\costw(\cd)$ gives:
$$\mue + \var \cd \total_Q- \var \cd \frac{\sum_{i \in Q} \ndraw_i^2}{\total_Q} + \mue + \var \cd \total_P - \var \cd \frac{\sum_{i \in P}\ndraw_i^2}{\total_P} $$
$$\geq \mue + \var \cd \total_Q + \var \cd (\total_Q - \ndraw_j) - \var \frac{\sum_{i \in Q} \ndraw_i^2 + \sum_{i \in P} \ndraw_i^2 - \ndraw_j^2}{\total_Q + \total_P - \ndraw_j} + \mue$$
Simplifying gives:
$$\var \cd \ndraw_j- \var \cd \frac{\sum_{i \in Q} \ndraw_i^2}{\total_Q}  - \var \cd \frac{\sum_{i \in P}\ndraw_i^2}{\total_P} \geq   - \var \frac{\sum_{i \in Q} \ndraw_i^2 + \sum_{i \in P} \ndraw_i^2 - \ndraw_j^2}{\total_Q + \total_P - \ndraw_j} $$
For convenience, we'll drop the common $\var$ coefficient as we continue simplifying: 
\begin{align*}
\ndraw_j- \frac{\sum_{i \in Q} \ndraw_i^2}{\total_Q}  - \frac{\sum_{i \in P}\ndraw_i^2}{\total_P} &\geq   - \frac{\sum_{i \in Q} \ndraw_i^2 + \sum_{i \in P} \ndraw_i^2 - \ndraw_j^2}{\total_Q + \total_P - \ndraw_j}\\
\ndraw_j &\geq   \frac{\sum_{i \in Q} \ndraw_i^2}{\total_Q}  + \frac{\sum_{i \in P}\ndraw_i^2}{\total_P} - \frac{\sum_{i \in Q} \ndraw_i^2 + \sum_{i \in P} \ndraw_i^2}{\total_Q + \total_P - \ndraw_j}  + \frac{\ndraw_j^2}{\total_Q + \total_P - \ndraw_j}\\
 \ndraw_j-\frac{\ndraw_j^2}{\total_Q + \total_P - \ndraw_j} &\geq   \frac{\sum_{i \in Q} \ndraw_i^2}{\total_Q}  + \frac{\sum_{i \in P}\ndraw_i^2}{\total_P} - \frac{\sum_{i \in Q} \ndraw_i^2 + \sum_{i \in P} \ndraw_i^2}{\total_Q + \total_P - \ndraw_j}\\
 \ndraw_j \cd \frac{\total_Q + \total_P - \ndraw_j - \ndraw_j}{\total_Q + \total_P - \ndraw_j} &\geq  \p{\sum_{i \in Q} \ndraw_i^2} \cd \p{\frac{1}{\total_Q} - \frac{1}{\total_Q + \total_P - \ndraw_j}} + \p{\sum_{i \in P} \ndraw_i^2} \cd \p{\frac{1}{\total_P} - \frac{1}{\total_Q + \total_P - \ndraw_j}}\\
 \ndraw_j \cd \frac{\total_Q + \total_P - 2\ndraw_j}{\total_Q + \total_P - \ndraw_j} &\geq  \p{\sum_{i \in Q} \ndraw_i^2} \cd \frac{\total_P - \ndraw_j}{\total_Q\cd (\total_Q + \total_P - \ndraw_j)} + \p{\sum_{i \in P} \ndraw_i^2} \cd \frac{\total_Q - \ndraw_j}{\total_P\cd (\total_Q + \total_P - \ndraw_j)}\\
 \ndraw_j \cd \p{\total_Q + \total_P - 2\ndraw_j} &\geq  \p{\sum_{i \in Q} \ndraw_i^2} \cd \frac{\total_P - \ndraw_j}{\total_Q} + \p{\sum_{i \in P} \ndraw_i^2} \cd \frac{\total_Q - \ndraw_j}{\total_P}\\
 \total_Q + \total_P - 2\ndraw_j &\geq  \p{\sum_{i \in Q} \frac{\ndraw_i^2}{\ndraw_j}} \cd \frac{\total_P - \ndraw_j}{\total_Q} + \p{\sum_{i \in P} \frac{\ndraw_i^2}{\ndraw_j}} \cd \frac{\total_Q - \ndraw_j}{\total_P}
\end{align*}
Because $\ndraw_j$ is the largest element, we can upper bound each term $\frac{\ndraw_i^2}{\ndraw_j}$ with $\ndraw_i$:
\begin{align*}
\total_Q + \total_P - 2\ndraw_j &\geq  \p{\total_Q} \cd \frac{\total_P - \ndraw_j}{\total_Q} + \p{\total_P} \cd \frac{\total_Q - \ndraw_j}{\total_P}  \\
\total_Q + \total_P - 2\ndraw_j &\geq \total_P - \total_j + \total_Q - \ndraw_j
\end{align*}
This gives an equality, and a strict inequality if $\ndraw_i < \ndraw_j$ for at least one player. Finally, we note that if the final structure is identical to the original structure, the cost is identical, so the inequality is similarly an equality.
\end{proof}

\gcolcore*
\begin{proof}
For reference, \cite{donahue2020model} analyzes a restricted example of $\ndraw_i \leq \frac{\mue}{\var}$ case, where players come in two types, $\ns, \nlv$, both $\leq \frac{\mue}{\var}$. Theorem 6.7 in that work shows that the grand coalition $\gcol$ is core stable for the two-type case. This lemma extends that result to show that $\gcol$ is core stable for the broader case of $\ndraw_i \leq \frac{\mue}{\var}$, where players may come in more than two sizes.

First, we will assume by contradiction that there exists a set $\colA \subset \col$, where $\col$ is the grand coalition, and where we assume that $err_j(\colA) < err_j(\col)$ for every $j \in \colA$. We will then show that this violates the requirement that $\ndraw_i \leq \frac{\mue}{\var}$ for all $i \in \col$, indicating that it is impossible for such a coalition $\colA$ to exist. 

By assumption, 
$$err_j(\col) > err_j(\colA)$$
Using $\total_\colA = \sum_{i \in \colA}\ndraw_i$ and $\total = \sum_{i \in \col}\ndraw_i$ we have:
$$\frac{\mue}{\total} + \var \cd \frac{\sum_{i\ne j}\ndraw_i^2 + (\total - \ndraw_j)^2}{\total^2} > \frac{\mue}{\total_\colA} + \var \cd \frac{\sum_{i\in \colA, i \ne j} \ndraw_i^2 + (\total_\colA - \ndraw_j)^2}{\total_\colA^2}$$
Multiplying each side by $\ndraw_j$ preserves the inequality: 
$$\frac{\mue}{\total}\cd \ndraw_j + \var\cd  \frac{\sum_{i\ne j}\ndraw_i^2 + (\total - \ndraw_j)^2}{\total^2}\cd \ndraw_j > \frac{\mue}{\total_\colA}\cd \ndraw_j + \var \cd \frac{\sum_{i\in \colA, i \ne j} \ndraw_i^2 + (\total_\colA - \ndraw_j)^2}{\total_\colA^2} \cd \ndraw_j$$
Next, we sum each side over all $j \in \colA$:
$$\sum_{j \in \colA}\cb{\frac{\mue}{\total}\cd \ndraw_j + \var\cd  \frac{\sum_{i\ne j}\ndraw_i^2 + (\total - \ndraw_j)^2}{\total^2}\cd \ndraw_j }> \sum_{j \in \colA} \cb{\frac{\mue}{\total_\colA}\cd \ndraw_j + \var \cd \frac{\sum_{i\in \colA, i \ne j} \ndraw_i^2 + (\total_\colA - \ndraw_j)^2}{\total_\colA^2} \cd \ndraw_j}$$
We will evaluate this sum term by term. The $\mue$ terms are simplest: 
\begin{align*}
\sum_{j \in \colA}\frac{\mue}{\total}\cd \ndraw_j &= \frac{\mue}{\total} \cd \total_\colA\\
\sum_{j \in \colA} \frac{\mue}{\total_\colA}\cd \ndraw_j  &= \mue
\end{align*}
For evaluating the sum of the $\var$ coefficient, we will first note that we can rewrite the numerator: 
$$\sum_{i\ne j}\ndraw_i^2 + (\total - \ndraw_j)^2 = \sum_{i\ne j}\ndraw_i^2 + \total^2 + \ndraw_j^2 - 2 \total \cd \ndraw_j = \sum_{i \in \col} \ndraw_i^2+ \total^2 - 2 \total \cd \ndraw_j$$
This means that the entire coefficient on the lefthand side can be rewritten as: 
\begin{align*}
 \sum_{j \in \colA}\cb{\var\cd  \frac{\sum_{i\ne j}\ndraw_i^2 + (\total - \ndraw_j)^2}{\total^2} \cd \ndraw_j} &= \sum_{j \in \colA}\cb{\var \cd \frac{\total^2 + \sum_{i \in \col}\ndraw_i^2 - 2 \total \cd \ndraw_j}{\total^2} \cd \ndraw_j}   \\
&=\sum_{j \in \colA}\cb{\var \cd \p{1 + \frac{\sum_{i \in \col} \ndraw_i^2}{\total^2} - 2 \frac{\ndraw_j}{\total}} \cd \ndraw_j } \\
&= \var \cd \total_\colA + \var \cd \total_\colA \frac{\sum_{i \in \col} \ndraw_i^2}{\total^2} - 2 \frac{\sum_{i \in \colA}\ndraw_i^2}{\total} 
\end{align*}
Similarly, we can rewrite the numerator of the $\var$ coefficient on the righthand side: 
$$\sum_{i\ne j, i \in \colA}\ndraw_i^2 + (\total_\colA - \ndraw_j)^2 = \sum_{i\ne j, i \in \colA}\ndraw_i^2 + \total_\colA^2 + \ndraw_j^2 - 2 \total_\colA \cd \ndraw_j = \sum_{i \in \colA} \ndraw_i^2+ \total_\colA^2 - 2 \total_\colA \cd \ndraw_j$$
Remember that $\colA \subset \col$. Similarly, we can rewrite the entire coefficient as: 
\begin{align*}
 \sum_{j \in \colA}\cb{\var\cd  \frac{\sum_{i\ne j, i \in \colA}\ndraw_i^2 + (\total_\colA - \ndraw_j)^2}{\total_\colA^2} \cd \ndraw_j }&= \sum_{j \in \colA}\cb{\var \cd \frac{\total_\colA^2 + \sum_{i \in \colA}\ndraw_i^2 - 2 \total_\colA \cd \ndraw_j}{\total_\colA^2} \cd \ndraw_j}\\
&=\sum_{j \in \colA}\cb{\var \cd \p{1 + \frac{\sum_{i \in \colA} \ndraw_i^2}{\total_\colA^2} - 2 \frac{\ndraw_j}{\total_\colA}} \cd \ndraw_j }\\
&= \var \cd \total_\colA + \var \cd \total_\colA \frac{\sum_{i \in \colA} \ndraw_i^2}{\total_\colA^2} - 2\cd \var \cd  \frac{\sum_{i \in \colA}\ndraw_i^2}{\total_\colA} \\
&=\var \cd \total_\colA + \var \cd \frac{\sum_{i \in \colA} \ndraw_i^2}{\total_\colA} - 2\cd \var \cd  \frac{\sum_{i \in \colA}\ndraw_i^2}{\total_\colA}\\
&= \var \cd \total_\colA - \var \cd  \frac{\sum_{i \in \colA}\ndraw_i^2}{\total_\colA}
\end{align*}
Combining these terms back into the inequality gives: 
$$\mue \cd \frac{\total_\colA}{\total} + \var \cd \total_\colA + \var \cd \frac{\total_{\colA}}{\total} \cd \frac{\sum_{i \in \col} \ndraw_i^2}{\total} - 2 \frac{\sum_{i \in \colA} \ndraw_i^2}{\total} > \mue + \var \cd \total_\colA - \var \frac{\sum_{i \in \colA} \ndraw_i^2}{\total_\colA}$$
Simplification: 
\begin{align*}
\mue \cd \frac{\total_\colA}{\total} + \var \cd \frac{\total_{\colA}}{\total} \cd \frac{\sum_{i \in \col} \ndraw_i^2}{\total} - 2 \frac{\sum_{i \in \colA} \ndraw_i^2}{\total} &> \mue - \var \frac{\sum_{i \in \colA} \ndraw_i^2}{\total_\colA}\\
\mue \cd \frac{\total_\colA}{\total} + \var \cd \frac{\total_{\colA}}{\total} \cd \frac{\sum_{i \in \col} \ndraw_i^2}{\total} - \var \frac{\sum_{i \in \colA} \ndraw_i^2}{\total} &> \mue - \var \frac{\sum_{i \in \colA} \ndraw_i^2}{\total_\colA} +  \var \frac{\sum_{i \in \colA} \ndraw_i^2}{\total}\\
\frac{\total_\colA}{\total} \cd \p{\mue + \var \cd \frac{\sum_{i \in \col} \ndraw_i^2}{\total} - \var \frac{\sum_{i \in \colA} \ndraw_i^2}{\total_\colA}} &> \mue +  \var \frac{\sum_{i \in \colA} \ndraw_i^2}{\total}- \var \frac{\sum_{i \in \colA} \ndraw_i^2}{\total_\colA}
\end{align*}
Note that the terms on the left and the right look very similar. We will strategically add and subtract a term on the left: 
$$\frac{\total_\colA}{\total} \cd \p{\mue + \var \cd \frac{\sum_{i \in \col} \ndraw_i^2 - \sum_{i \in \colA} \ndraw_i^2 + \sum_{i \in \colA}\ndraw_i^2}{\total} - \var \frac{\sum_{i \in \colA} \ndraw_i^2}{\total_\colA}} > \mue +  \var \frac{\sum_{i \in \colA} \ndraw_i^2}{\total}- \var \frac{\sum_{i \in \colA} \ndraw_i^2}{\total_\colA} $$
Multiplying on the left side: 
$$\frac{\total_\colA}{\total} \cd \var \cd \frac{\sum_{i \in \col \setminus \colA} \ndraw_i^2}{\total} + \frac{\total_\colA}{\total} \cd \p{\mue + \var \cd \frac{\sum_{i \in \colA}\ndraw_i^2}{\total} - \var \frac{\sum_{i \in \colA} \ndraw_i^2}{\total_\colA}} > \mue +  \var \frac{\sum_{i \in \colA} \ndraw_i^2}{\total}- \var \frac{\sum_{i \in \colA} \ndraw_i^2}{\total_\colA}$$
Collecting terms: 
$$\frac{\total_\colA}{\total} \cd \var \cd \frac{\sum_{i \in \col \setminus \colA} \ndraw_i^2}{\total} + \frac{\total_\colA}{\total} \cd \p{\mue +  \var \cd \p{\sum_{i \in \colA}\ndraw_i^2}\cd \p{\frac{1}{\total}- \frac{1}{\total_\colA}}} > \mue +  \var \cd \p{\sum_{i \in \colA}\ndraw_i^2}\cd \p{\frac{1}{\total}- \frac{1}{\total_\colA}}$$
Changing the sign:
$$\frac{\total_\colA}{\total} \cd \var \cd \frac{\sum_{i \in \col \setminus \colA} \ndraw_i^2}{\total} + \frac{\total_\colA}{\total} \cd \p{\mue -  \var \cd \p{\sum_{i \in \colA}\ndraw_i^2}\cd \p{\frac{1}{\total_\colA}- \frac{1}{\total}}} > \mue -  \var \cd \p{\sum_{i \in \colA}\ndraw_i^2}\cd \p{\frac{1}{\total_\colA}- \frac{1}{\total}}$$
Bringing across terms to the righthand side: 
$$\frac{\total_\colA}{\total} \cd \var \cd \frac{\sum_{i \in \col \setminus \colA} \ndraw_i^2}{\total} > \p{1 - \frac{\total_\colA}{\total}} \cd \p{\mue -  \var \cd \p{\sum_{i \in \colA}\ndraw_i^2}\cd \p{\frac{1}{\total_\colA}- \frac{1}{\total}}}$$
Bringing all coefficients of $\var$ to the lefthand side:
$$\frac{\total_\colA}{\total} \cd \var \cd \frac{\sum_{i \in \col \setminus \colA} \ndraw_i^2}{\total}+ \p{1 - \frac{\total_\colA}{\total}} \cd \var \cd \p{\sum_{i \in \colA}\ndraw_i^2}\cd \p{\frac{1}{\total_\colA}- \frac{1}{\total}} > \p{1 - \frac{\total_\colA}{\total}} \cd \mue$$
Rewriting: 
$$\frac{\total_\colA}{\total} \cd \var \cd \p{\sum_{i \in \col \setminus \colA} \ndraw_i^2}+ \p{\total - \total_\colA} \cd \var \cd \p{\sum_{i \in \colA}\ndraw_i^2}\cd \p{\frac{1}{\total_\colA}- \frac{1}{\total}} > \p{\total - \total_\colA} \cd \mue$$
We strategically rewrite the righthand side: 
$$\frac{\total_\colA}{\total} \cd \var \cd \p{\sum_{i \in \col \setminus \colA} \ndraw_i^2}+ \p{\total - \total_\colA} \cd \var \cd \p{\sum_{i \in \colA}\ndraw_i^2}\cd \p{\frac{1}{\total_\colA}- \frac{1}{\total}} > (\total - \total_\colA) \cd \mue \cd \frac{\total_\colA}{\total} + (\total - \total_\colA)\cd \p{1 - \frac{\total_\colA}{\total}}\cd \mue$$
$$\frac{\total_\colA}{\total} \cd \var \cd \p{\sum_{i \in \col \setminus \colA} \ndraw_i^2}+ \p{\total - \total_\colA} \cd \var \cd \p{\sum_{i \in \colA}\ndraw_i^2}\cd \p{\frac{1}{\total_\colA}- \frac{1}{\total}}> (\total - \total_\colA) \cd \mue \cd \frac{\total_\colA}{\total} + \p{\total - \total_\colA} \cd \total_\colA\cd \p{\frac{1}{\total_\colA}- \frac{1}{\total}}\cd \mue$$
We pull all of the terms over to the lefthand side: 
$$\frac{\total_\colA}{\total} \cd \p{\sum_{i \in \col \setminus \colA} \ndraw_i \cd \p{\var\cd \ndraw_i - \mue}}+ \p{\total - \total_\colA} \cd \p{\frac{1}{\total_\colA}- \frac{1}{\total}}\cd \p{\sum_{i \in \colA}\ndraw_i \cd \p{\ndraw_i \cd \var - \mue}}  > 0$$
Finally, we will show that the above inequality cannot hold. By assumption, $\ndraw_i \leq \frac{\mue}{\var}$ for all $i \in \col$. This means that $\var \cd \ndraw_i - \mue$ is negative for all $i \in \col$. Because every other term on the lefthand side is positive (note that $\frac{1}{\total_\colA} > \frac{1}{\total}$), we know that the lefthand term is negative. However, the inequality is requiring that the term is positive. By this contradiction, we know that the initial assumption must have been wrong: so long as $\ndraw_i \leq \frac{\mue}{\var}$, there cannot be any set $\colA$ such that each player strictly prefers $\colA$ to $\col$, so the grand coalition $\col$ is core stable. 
\end{proof}

\section{Price of Anarchy}\label{app:poa}

\alonecore*
\begin{proof}
By Lemma 5.3 in \cite{donahue2020model}, when all players have $\geq \frac{\mue}{\var}$ samples, each player with size $> \frac{\mue}{\var}$ minimizes its error by doing local learning. By the same lemma, each player of size exactly equal to $\frac{\mue}{\var}$ minimize their error in any arrangement with other players also of size $\frac{\mue}{\var}$. Taken together, this implies that the only stable arrangements are ones where all players of size $> \frac{\mue}{\var}$ are doing local learning and all players of size equal $\frac{\mue}{\var}$ are arranged in any grouping. Because all of these have equal error to the minimal error, the Price of Anarchy is equal to 1.
\end{proof}

\PoA*

\begin{proof}
This theorem is the result of multiple lemmas, each of which handle players of different sizes in different situations. Theorem \ref{tab:poatable} summarizes these contributions. Specifically, it divides players into four different types ($T_0, T_1, T_2, T_3$) based on their size and the group they are federating with in $\partition_{M}$. These results are summarized in Table \ref{tab:poatable} and described below. 

First, we note that by Lemma \ref{lem:betterthanalone} the highest error any player can experience in $\partition_{M}$ is $\frac{\mue}{\ndraw_i}$, so the cost due to a particular player in $\partition_{M}$ is upper bounded by $\mue$. 

\begin{table}[]
\begin{tabular}{|C{1cm}|C{5cm}|C{3.5cm}|C{4cm}|}
\hline
\textbf{Type} & \textbf{Condition}                                                                                                                                                                       & \textbf{Upper bound on $err_i(\partition_{M})$}                                                & \textbf{Lower bound on $err_i(\partition_{opt})$}                                       \\ \hline
$T_0$         & $\ndraw_i \geq \frac{\mue+\var}{2\cd \var}$                                                                                                                                              & \multirow{2}{*}{$\frac{\mue}{\ndraw_i}$, by Lemma \ref{lem:betterthanalone}.} & $\frac{1}{2} \frac{\mue}{\ndraw_i}$, by Lemma \ref{lem:lowerbounderror} \\ \cline{1-2} \cline{4-4} 
$T_1$         & $\frac{\mue}{9\cd \var} \leq \ndraw_i \leq \frac{\mue+\var}{2\var}$                                                                                                                      &                                                                                                & \multirow{3}{*}{$\var$, by Lemma \ref{lem:lowerbounderror}}             \\ \cline{1-3}
$T_2$         & $\ndraw_i < \frac{\mue}{9 \cd \var}$ and is federating with other players of total mass at least $\frac{\mue}{3 \var}$ in $\partition_M$.                                                & $7.25 \cd \var$, by Lemma \ref{lem:twcspec}                                   &                                                                                          \\ \cline{1-3}
$T_3$         & $\ndraw_i < \frac{\mue}{9 \cd \var}$ and is NOT federating with other players of total mass at least $\frac{\mue}{3 \var}$ in $\partition_M$. & Unbounded, but Lemma \ref{lem:relaxed} gives a stability result.               &                                                                                          \\ \hline
\end{tabular}
\caption{Summary of relevant bounds for proof of Theorem \ref{thrm:PoA}.}
\label{tab:poatable}
\end{table}

\begin{itemize}
    \item Say that player $i \in T_0$ if $\ndraw_i \geq \frac{\mue+\var}{2\var}$. Lemma \ref{lem:lowerbounderror} shows that if $\ndraw_i \geq \frac{\mue + \var}{2\var}$, then $err_i(\partition_{opt}) \geq \frac{1}{2} \frac{\mue}{\ndraw_i}$, so player $i$'s contribution to the weighted cost is  $ \geq \frac{1}{2} \cd \mue$. 
    \item Say that player $i \in T_1$ if  $\frac{\mue}{9\cd \var} \leq \ndraw_i \leq \frac{\mue+\var}{2\var}$. Lemma \ref{lem:lowerbounderror} shows that $err_i(\partition_{opt}) \geq \var $ for $\ndraw_i \leq \frac{\mue + \var}{2\var}$, so player $i$'s contribution to the weighted cost is  $ \geq \var \cd \ndraw_i$.
    \item Say that player $i \in T_2$ if $\ndraw_i < \frac{\mue}{9 \cd \var}$ and if, in $\partition_M$, it is federating with other players of total mass at least $\frac{\mue}{3 \var}$. Then, by Lemma \ref{lem:twcspec} $err_i(\partition_{M}) \leq 7.25 \var \leq 7.5 \cd \var$. Lemma \ref{lem:lowerbounderror} applies again and shows that $err_i(\partition_{opt}) \geq \var $ for $\ndraw_i \leq \frac{\mue + \var}{2\var}$, so player $i$'s contribution to the weighted cost is  $ \geq \var \cd \ndraw_i$.
    \item Say $i \in T_3$ if $\ndraw_i \leq \frac{\mue}{9 \cd \var}$ and if in $\partition_{M}$ it is \emph{not} federating with other players of total mass at least $\frac{\mue}{3\cd \var}$. Then, by Lemma \ref{lem:relaxed} there is at most one group of such description in $\partition_M$ (or any IS arrangement) - call it $A$. What is this group's total contribution to the cost? 
    $$\mue + \var \cd \total_{A} - \var \frac{\sum_{i \in A} \ndraw_i^2}{\total_{A}} \leq \mue + \var \cd \total_{A} - \var \frac{\total_A}{\total_{A}} \leq^* \p{1 + \frac{1}{3} + \frac{1}{9}}\mue - \var < 1.5\mue$$
    where in the step marked with $*$ we have upper bounded $\total_T$ by the knowledge that it contains a player of size $\leq \frac{\mue}{9\var}$ is federating with partners of total size no more than $\frac{\mue}{3\var}$. Note that $\total_T$ is the mass of the entire group containing $T_3$ players, and so may double-count the contributions of some players not in $T_3$. 
\end{itemize}
Next, we bring these terms together to bound the overall result. Note that $\costw(\partition)$ is a weighted cost that is obtained by multiplying player $j$'s error by its number of samples $\ndraw_j$. 
$$PoA = \frac{\costw(\partition_{M})}{\costw(\partition_{opt})} \leq \frac{\vert T_0 \vert \cd \mue + \vert T_1 \vert \cd \mue + \sum_{i \in T_2} 7.5 \cd \var \cd \ndraw_i + 1.5 \mue }{\vert T_0 \vert \cd \frac{\mue}{2} + \sum_{i \in T_1} \var \cd \ndraw_i  + \sum_{i \in T_2}  \var \cd \ndraw_i + \sum_{i \in T_3} \var \cd \ndraw_i}$$
First, we note that if there do not exist any players in $T_3$, then we can write the bound as: 
$$\frac{\vert T_0 \vert \cd \mue + \vert T_1 \vert \cd \mue + \sum_{i \in T_2}7.5 \cd  \var \cd \ndraw_i}{\vert T_0 \vert \cd \frac{\mue}{2} + \vert T_1 \vert \cd \frac{\mue}{9}  + \sum_{i \in T_2}  \var \cd \ndraw_i } \leq 9$$
Suppose that $\vert T_3 \vert \geq 1$. Then, the main goal is to absorb the additive $1.5 \cd  \mue $ term.

First, we consider the case where we have some player $\ndraw_j \geq  \frac{\mue}{3\var}$, which we will show implies a PoA bound of 9. Any player of size $\geq \frac{\mue}{3\cd \var}$ must be in $T_0$ or  $T_1$. First, we will assume that $j \in T_0$, so $\vert T_0 \vert \geq 1$, meaning:  
$$4.5 \cd \vert T_0 \vert \cd \mue \geq \vert T_0 \vert \cd \mue + 1.5 \cd \mue $$
This means the bound can be upper bounded by: 
$$PoA \leq  \frac{4.5\vert T_0 \vert \cd \mue + \vert T_1 \vert \cd \mue + \sum_{i \in T_2} 7.5 \cd \var \cd \ndraw_i}{\vert T_0 \vert \cd \frac{\mue}{2} + \vert T_1 \vert \cd \frac{\mue}{9\var}  + \sum_{i \in T_2}  \var \cd \ndraw_i} \leq 9$$
Next, we consider the case where $j \in T_1$ and $\vert T_0 \vert =0$. Then, the upper bound becomes: 
\begin{align*}
PoA&<\frac{(\vert T_1\vert -1)\cd \mue + \mue + 7.5 \var \cd \sum_{i \in T_2} \ndraw_i + 1.5 \mue}{\sum_{i \ne j, i \in T_1} \var \cd \ndraw_i + \var \cd \ndraw_j + \var \cd \sum_{i \in T_2} \ndraw_i}\\
&<\frac{(\vert T_1\vert -1)\cd \mue + \mue + 7.5 \var \cd \sum_{i \in T_2} \ndraw_i + 1.5 \mue}{(\vert T_1 \vert -1)\cd \frac{\mue}{9} + \frac{\mue}{3} + \var \cd \sum_{i \in T_2} \ndraw_i}   \\
&<\frac{(\vert T_1\vert -1)\cd \mue + 2.5\mue + 9 \var \cd \sum_{i \in T_2} \ndraw_i}{(\vert T_1 \vert -1)\cd \frac{\mue}{9} + \frac{\mue}{3} + \var \cd \sum_{i \in T_2} \ndraw_i}\\
&<\frac{(\vert T_1\vert -1)\cd \mue + 3\mue + 9 \var \cd \sum_{i \in T_2} \ndraw_i}{(\vert T_1 \vert -1)\cd \frac{\mue}{9} + \frac{\mue}{3} + \var \cd \sum_{i \in T_2} \ndraw_i} \\
&= 9
\end{align*}
Finally, we consider the case where all players have size $\leq \frac{\mue}{3\var}$. By Lemma \ref{lem:relaxed}, if there exist any players in $T_3$, then the entire arrangement is only stable if $\partition_M = \gcol = \partition_{opt}$, giving a PoA of 1. 

These proofs taken together show that the overall PoA is upper bounded by 9. 
\end{proof}

\lowerbounderror*

\begin{proof}
Player $j$'s error when federating with the coalition $\col$ is: 
$$err_j(\col \cup \{\ndraw_j\}) = \frac{\mue}{\total_\col + \ndraw_j} + \var \frac{\sum_{i \in \col}\ndraw_i^2 + \total_\col^2}{(\total_\col + \ndraw_j)^2}$$
Given a fixed $\total_{\col}$, $\sum_{i \in \col} \ndraw_i^2$ is minimized when all of the players besides $j$ have size $\ndraw_i = \frac{\total_\col}{\vert \col \vert}$, which means that  $\ndraw_i^2 = \frac{\total_\col^2}{\vert \col \vert^2}$. The error is thus lower bounded by: 
$$err_j(\col \cup \{\ndraw_j\}) \geq \frac{\mue}{\total_\col + \ndraw_j} + \var \frac{\frac{\total_\col^2}{\vert \col \vert} + \total_\col^2}{(\total_\col + \ndraw_j)^2}$$
This decreases with $\vert \col \vert$, so we set $\vert \col \vert = \total_\col$ to further lower bound the error: 
$$\geq \frac{\mue}{\total_\col + \ndraw_j} + \var \frac{\total_\col + \total_\col^2}{(\total_\col + \ndraw_j)^2}$$
Note that the \enquote{units} of this term might seem strange: the numerator of the $\var$ component involves a $\total_{\col}$ and $\total_{\col}^2$. This is because we assumed that $\sum_{i \in \col} \ndraw_i^2 \geq \total_{\col}$, which is correct in magnitude but which involves different units. 

Next, we will lower bound this term by analyzing how it changes with $\total_{\col}$. First, we take the derivative with respect to $\total_\col$: 
$$\frac{\ndraw_j \cd (\var - \mue + 2 \var \cd \total_{\col})  -\total_\col \cd (\mue+ \var)}{(\total_\col + \ndraw_j)^3}$$
\textbf{Case 1: Derivative always negative}\\
In some situations, this derivative is always negative (the player $j$ always prefers $\total_\col$ as large as possible). When does this occur? 
$$\ndraw_j \cd (\var - \mue + 2\total_\col \cd \var) < (\mue + \var) \cd \total_\col \quad \forall\total_{\col}$$
As $\total_{\col} \rightarrow \infty$, the $\var - \mue$ additive term on the lefthand side becomes irrelevant, so what we require is 
\begin{align*}
2 \var \cd \ndraw_j\cd \total_{\col} & \leq (\mue + \var) \cd \total_{\col}\\
\ndraw_j  &\leq \frac{\mue + \var}{2 \var}
\end{align*}
For players satisfying this premise, we can lower bound their error by sending $\total_\col \rightarrow \infty$ in the original error equation. 
$$\lim_{\total_{\col}\rightarrow \infty}\br{\frac{\mue}{\total_\col + \ndraw_j} + \var \frac{\total_\col + \total_\col^2}{(\total_\col + \ndraw_j)^2}} = \var $$
This implies that the player's error goes to $\var$ (from above), so is lower bounded by $\var$. \\
\textbf{Case 2: Derivative sometimes negative, sometimes positive}\\
Next, we'll consider the case where $\ndraw_j >\frac{\mue + \var}{2 \var}$. The second derivative of the player's error with respect to $\total_{\col}$ is: 
$$2 \cd \var \cd \ndraw_j - \mue - \var$$
which is greater than or equal to 0 in this case. In order to lower bound the overall error, we must bound the error when $\total_{\col}=0$ (at its minimum value) and when the derivative with respect to $\total_{\col}$ is 0 (local minimum). Note that when $\total_{\col} = 0$, player $j$'s error is $\frac{\mue}{\ndraw_j}$, which is $> \frac{1}{2} \cd \frac{\mue}{\ndraw_j}$, satisfying the premise. Next, we will consider the case where the derivative is equal to 0: In this case, the slope isn't always negative, so there must be some $\total_\col$ such that the slope is equal to 0. This occurs when: 
\begin{align*}
\ndraw_j \cd (\var - \mue) + \total_\col \cd (2 \ndraw_j \cd \var -\mue - \var)&=0\\  
\total_\col &= \frac{\ndraw_j \cd (\mue - \var)}{2 \ndraw_j \cd \var - \mue - \var}
\end{align*}
Substituting in for this value of $\total_{\col}$ into player $j$'s error gives: 
$$\frac{-\mue^2 - 2 \mue \cd \var +4 \ndraw_j \cd \mue \cd \var -(\var)^2}{-4 \ndraw_j \cd \var +4 \ndraw_j^2 \cd \var} = \frac{\mue}{\ndraw_j} \cd \frac{-\mue- 2 \var +4 \ndraw_j \cd  \var -\var \cd \frac{\var}{\mue}}{-4 \cd \var +4 \ndraw_j \cd \var}$$
In order to prove that this whole term is lower bounded by $\frac{1}{2} \frac{\mue}{\ndraw_j}$, we will show that the coefficient on $\frac{\mue}{\ndraw_j}$ is lower bounded by $\frac{1}{2}$. Because $\ndraw_j\geq 1$, we know that the denominator is positive:
\begin{align*}
\frac{-\mue- 2 \var +4 \ndraw_j \cd  \var -\var \cd \frac{\var}{\mue}}{-4 \cd \var +4 \ndraw_j \cd \var} &\geq \frac{1}{2}  \\
-2 \mue -4 \var  + 8 \ndraw_j \cd \var - \frac{\sigma^4}{\mue} &\geq -4 \var + 4 \ndraw_j \cd \var \\
-2 \mue + 4 \ndraw_j \cd \var - \frac{\sigma^4}{\mue} &\geq 0\\
\ndraw_j &\geq \frac{\mue}{2\var} + \frac{\var}{4\mue}
\end{align*}
This is satisfied if the lower bound is smaller than or equal to $\frac{\mue + \var}{2\var}$. We can show this by noting that $\mue \geq \var$ for any avenue of interest (otherwise, $\frac{\mue}{\var} < 1$ and by Lemma \ref{lem:alonecore} the only stable arrangement is to have all players doing local learning). This means that: 
$$\frac{\mue}{2\var} + \frac{\var}{4\mue}\leq \frac{\mue}{2\var} + \frac{1}{4}  = \frac{\mue + \frac{1}{2} \var}{2\var} < \frac{\mue + \var}{2\var}$$
as desired. This shows that: 
$$err_j(\col \cup \{\ndraw_j\}) \geq \frac{1}{2}\frac{\mue}{\ndraw_j} $$
\end{proof}

\twcspec*

\begin{proof}
The error a player $\ndraw_j$ experiences is given by: 
$$err_j(\col \cup \{\ndraw_j\}) = \frac{\mue}{\ndraw_j + \total_{\col}} + \var \frac{\sum_{i \in l} \ndraw_i^2 + \total_{\col}^2}{(\total_{\col} + \ndraw_j)^2}$$
Given a fixed total sum $\total_{\col}$, the $\sum_{i \in l} \ndraw_i^2$ term is maximized when all of the mass is on a single partner. So the overall cost can be upper bounded by: 
$$<\frac{\mue}{\total_{\col} + \ndraw_j} + \var \frac{2\total_{\col}^2}{(\total_{\col} + \ndraw_j)^2}$$
Taking the derivative with respect to $\total_{\col}$ gives: 
\begin{align*}
 -\frac{\mue}{(\total_{\col} +\ndraw_j)^2} + \var \frac{4 \total_{\col} \cd (\total_{\col} + \ndraw_j)^2 - 4 \total_{\col}^2 \cd (\total_{\col} + \ndraw_j)}{(\total_{\col} + \ndraw_j)^4}&=-\frac{\mue}{(\total_{\col} +\ndraw_j)^2} + \var \frac{4 \total_{\col} \cd \ndraw_j}{(\total_{\col} + \ndraw_j)^3} \\
&= \frac{- \mue \cd (\total_{\col} + \ndraw_j)+ 4\var  \total_{\col} \cd \ndraw_j}{(\total_{\col} + \ndraw_j)^3}
\end{align*}

Next, we will upper bound player $j$'s error based on the sign of the derivative with respect to $\total_{\col}$. 

\textbf{Case 1: Derivative with respect to $\total_{\col}$ always positive}: \\
This occurs when the numerator is positive for all $\total_{\col} \geq 0$, or
\begin{align*}
- \mue \cd (\total_{\col} + \ndraw_j)+ 4\var  \total_{\col} \cd \ndraw_j &> 0 \\
\total_{\col} \cd (4 \var \cd \ndraw_j - \mue) &> \mue \cd \ndraw_j 
\end{align*}
To begin with, we must have that $4 \var \cd \ndraw_j > \mue$ or else the lefthand side is negative, so $\ndraw_j > \frac{\mue}{4 \var}$. Given that, the error is largest when $\total_{\col}$ is set to its largest value of $\frac{\mue}{3\var}$. 
\begin{align*}
\frac{\mue}{3\cd \var} \cd (4 \var \ndraw_j - \mue) &> \mue \cd \ndraw_j\\
4 \var \ndraw_j - \mue &> 3 \var \cd  \ndraw_j\\
\ndraw_j &> \frac{\mue}{\var}
\end{align*}
If this is the case, what is the maximum amount of error that $\ndraw_j$ receives? The error is in the form:
$$\frac{\mue}{\total_{\col} + \ndraw_j} + \var \frac{2\total_{\col}^2}{(\total_{\col} + \ndraw_j)^2}$$
We know that this is maximized when $\total_{\col} \rightarrow \infty$. In this case, $\mue$ term goes to 0. The $\var$ term (by L'H\^{o}pital's rule) goes to: 
$$\var \frac{4\total_{\col}}{2 (\total_{\col} + \ndraw_j)} \rightarrow  2\var  $$
\textbf{Case 2: Derivative with respect to $\total_{\col}$ is always negative} \\
Next, we'll consider the inverse case where the derivative is always negative. This occurs when: 
$$\total_{\col} \cd (4 \var \cd \ndraw_j - \mue)< \mue \cd \ndraw_j \quad  \forall \total_{\col}$$
This has to be true for all $\total_{\col}$, which implies that the $4 \var \cd \ndraw_j - \mue$ term is negative, or $\ndraw_j \leq \frac{\mue}{4\var}$. If this is the case, the maximal error is achieved when the $\total_{\col}$ term is smallest ($\frac{\mue}{3\cd  \var}$). Plugging into the error form gives us: 
\begin{align*}
\frac{\mue}{\frac{\mue}{3 \cd \var} + \ndraw_j} + \var \frac{2 \cd  \frac{\mue^2}{9\cd \sigma^4}}{\p{\frac{\mue}{3 \cd \var} + \ndraw_j}^2}&=\frac{ \mue \cd \p{\frac{\mue}{3 \cd \var} + \ndraw_j} + \frac{2 \mue^2}{9\cd \var}}{\p{\frac{\mue}{3 \cd \var} + \ndraw_j}^2} \\
&=\frac{ \mue \cd \p{\frac{\mue}{3 \cd \var} + \ndraw_j} + \frac{2 \mue^2}{9\cd \var}}{\frac{1}{9 \sigma^4}\cd \p{\mue + 3 \var \cd \ndraw_j}^2}\\
&<\frac{ 9 \sigma^4 \mue \cd \p{\frac{\mue}{3 \cd \var} + \ndraw_j} + 2 \cd \mue^2 \cd \var}{\mue^2}\\
&= 3\var + 9 \var \cd \frac{\var}{\mue} \cd \ndraw_j + 2 \var \\
&< 5 \var + 9 \var \cd \frac{\var}{\mue} \cd \frac{\mue}{4\var} \\
& = 7.25
\end{align*}
where in the last step we have used that $\ndraw_j \leq \frac{\mue}{4\var}$. \\
\textbf{Case 3: when the derivative with respect to $\total_{\col}$ is sometimes positive and sometimes negative}\\
Using the values above, we know this occurs when $\frac{\mue}{4\var} \leq \ndraw_j \leq \frac{\mue}{\var}$.  First, we'll confirm that the error first decreases and then increases with $\total_{\col}$. The derivative is: 
$$\total_{\col} \cd (4 \var \cd \ndraw_j - \mue) - \mue \cd\ndraw_j$$
Here, we are assuming that the coefficient on $\total_{\col}$ is either 0 or positive, so the second derivative with respect to $\total_{\col}$ is positive. Given that the derivative is negative at some point, it must be negative for small $\total_{\col}$.  We know from Case 1 that as $\total_{\col} \rightarrow \infty$, the error goes to $2\var$, so in order to bound the entire space, we only need to bound the error at the smallest value of $\total_{\col}$, which is $\frac{\mue}{3\cd \var}$. The first few steps are identical to Case 2: 
$$\frac{\mue}{\frac{\mue}{3 \cd \var} + \ndraw_j} + \var \frac{2 \cd  \frac{\mue^2}{9\cd \sigma^4}}{\p{\frac{\mue}{3 \cd \var} + \ndraw_j}^2}=\frac{ \mue \cd \p{\frac{\mue}{3 \cd \var} + \ndraw_j} + \frac{2 \mue^2}{9\cd \var}}{\p{\frac{\mue}{3 \cd \var} + \ndraw_j}^2} =\frac{ \mue \cd \p{\frac{\mue}{3 \cd \var} + \ndraw_j} + \frac{2 \mue^2}{9\cd \var}}{\frac{1}{9 \sigma^4}\cd \p{\mue + 3 \var \cd \ndraw_j}^2}$$
In the next step, though, we use that $\frac{\mue}{4\var} \leq \ndraw_j \leq \frac{\mue}{\var}$. 
\begin{align*}
&<\frac{ 9 \sigma^4 \mue \cd \p{\frac{\mue}{3 \cd \var} + \ndraw_j} + 2 \cd \mue^2 \cd \var}{(\mue + \frac{3}{4} \mue)^2}\\
&= \frac{3\var + 9 \var \cd \frac{\var}{\mue} \cd \ndraw_j + 2 \var}{\frac{49}{16}}  \\
&< \frac{16}{49} \cd \p{5 \var + 9 \var \cd \frac{\var}{\mue} \cd \frac{\mue}{\var}}\\
&= \frac{16}{49} \cd 14 \cd \var \\
& < 5 \var
\end{align*}
Of the three cases, the highest bound is $7.25 \cd\var$. 
\end{proof}

Lemma \ref{lem:relaxed}, below, relies on Lemmas \ref{lem:welcome}, \ref{lem:case12}, and \ref{lem:case3}, which are stated and proved immediately after the proof of Lemma \ref{lem:relaxed}. 
\relaxed*
\begin{proof}
By Lemma \ref{lem:welcome}, we know that every player in every group welcomes the addition of any other player. Therefore, in order to prove that this arrangement isn't individually stable, we simply have to prove that a player would wish to move.

We will consider a cluster $A$ with elements $i \in T_3$ present. We know that there exists at least one element in $A$ s.t. the mass of its partners ($\total_A - \ndraw_i$) is less than $\frac{\mue}{3\var}$. This implies also that $\total - \ndraw_a < \frac{\mue}{3\var}$ for $\ndraw_a$ the largest element in $A$. We also know that $\ndraw_a < \frac{\mue}{3\var}$ because we know that there exists some other element in the cluster with $\total_A - \ndraw_i < \frac{\mue}{3\var}$. 

Next, let's suppose there exists some other cluster $B$, such that all elements are $\leq \frac{\mue}{3\var}$ in size. We will consider some $\ndraw_b$ largest player in $B$. There are four possible cases:  

\begin{enumerate}
    \item $\ndraw_a \geq \ndraw_b, \total_A - \ndraw_a \geq \total_B - \ndraw_b$: Unstable by Lemma \ref{lem:case12} (player $b$ wishes to move to $A$). 
    \item (Symmetric to above) $\ndraw_a \leq \ndraw_b$, $\total_A - \ndraw_a  \leq \total_B - \ndraw_b$: Unstable by Lemma \ref{lem:case12} (player $a$ wishes to move to $B$). 
    \item $\ndraw_a > \ndraw_b$, $\total_A - \ndraw_a < \total_B - \ndraw_b$. Note that in this case, we know that $\total_A - \ndraw_a \leq \frac{\mue}{3\var}$, so we satisfy the conditions of Lemma \ref{lem:case3}, and thus player $a$ would prefer to join $B$. 
    \item $\ndraw_a < \ndraw_b$, $\total_A- \ndraw_a > \total_B - \ndraw_b$. In this case, we know that $\frac{\mue}{3\var} > \total_A - \ndraw_a > \total_B -\ndraw_b$, so we again satisfy the conditions of Lemma \ref{lem:case3}, and thus player $b$ would prefer to join $A$.
    \end{enumerate}
\end{proof}

\begin{lemma}\label{lem:welcome}
A group of players where each has size $\ndraw_i\leq \frac{\mue}{3\var}$ always welcomes the addition of another player of size $\ndraw_k \leq \frac{\mue}{3\var}$.  
\end{lemma}
\begin{proof}
For this section, we will rewrite the form of the error that a player experiences while federating with a coalition $\col$. Specifically, we will write the error in the form below, where $a_i$ refers to the number of players with number of samples $\ndraw_i$. 
$$\frac{\mue}{\sum_{i=1}^{\nplayer}a_i \cd \ndraw_i} + \var \frac{\sum_{i\ne j}a_i \cd \ndraw_i^2 + (a_j-1) \cd \ndraw_j^2 + (\sum_{i\ne j} a_i \cd \ndraw_i + (a_j-1) \cd \ndraw_j)^2}{(\sum_{i=1}^{\nplayer}a_i \cd \ndraw_i)^2}$$
Setting $\total = \sum_{i=1}^{\nplayer}a_i \cd \ndraw_i$ gives: 
$$\frac{\mue}{\total} + \var \cd \frac{\sum_{i\ne j}a_i \cd \ndraw_i^2 + (a_j-1) \cd \ndraw_j^2 + (\total - \ndraw_j)^2}{\total^2}$$
In order to prove that any player $j$ welcomes the addition of any other player $k$, we will show that the derivative with respect to $a_k$ is always negative. This means that player $j$ always sees its error decrease with the addition of another player of size $\ndraw_k$. 
As we take the derivative, the coefficient on the $\mue$  term in the error value becomes: 
$$-\frac{\mue \cd \ndraw_k}{\total^2} = -\frac{\mue \cd \ndraw_k \cd \total^2}{\total^4}$$
The derivative of the coefficient on the $\var$ term becomes: 
$$\frac{\var}{\total^4} \cd \p{ \p{\ndraw_k^2 + 2 (\total - \ndraw_j) \cd \ndraw_k} \cd \total^2 - \p{\sum_{i\ne j} a_i \cd \ndraw_i^2 + (a_j - 1) \cd \ndraw_j^2 + (\total - \ndraw_j)^2} \cd 2 \cd \total \cd \ndraw_k}$$
So, the overall derivative is negative if: 
$$\mue \cd \ndraw_k \cd \total^2 > \var \cd \p{ \p{\ndraw_k^2 + 2 (\total - \ndraw_j) \cd \ndraw_k} \cd \total^2 - \p{\sum_{i\ne j} a_i \cd \ndraw_i^2 + (a_j - 1) \cd \ndraw_j^2 + (\total - \ndraw_j)^2} \cd 2 \cd \total \cd \ndraw_k}$$
We pull out and cancel common terms: 
\begin{align*}
\mue \cd \ndraw_k \cd \total^2 &> \var \cd \ndraw_k \cd \total \cd \p{\p{\ndraw_k + 2 \total - 2 \ndraw_j} \cd \total - 2\p{\sum_{i\ne j} a_i \cd \ndraw_i^2 + (a_j - 1) \cd \ndraw_j^2 + (\total - \ndraw_j)^2}}\\
\mue \cd \total & > \var \cd \p{\p{\ndraw_k + 2 (\total - \ndraw_j)} \cd (\total-\ndraw_j + \ndraw_j) - 2\p{\sum_{i\ne j} a_i \cd \ndraw_i^2 + (a_j - 1) \cd \ndraw_j^2 + (\total - \ndraw_j)^2}}
\end{align*}
Strategically expanding: 
$$\mue \cd \total > \var \cd \p{\ndraw_k \cd \total + 2 (\total -\ndraw_j)^2 + 2 \ndraw_j \cd \total - 2\ndraw_j^2 - 2 \p{\sum_{i\ne j} a_i \cd \ndraw_i^2 + (a_j-1) \cd \ndraw_j^2} - 2 (\total - \ndraw_j)^2}$$
Collecting: 
$$\mue \cd \total > \var \cd \p{\total \cd (\ndraw_k + 2 \ndraw_j) -2 \sum_{i=1}^{\nplayer} a_i \cd \ndraw_i^2}$$
Substituting in for $\total$:
\begin{align*}
\mue \cd \sum_{i=1}^{\nplayer} a_i \cd \ndraw_i&> \var \cd \p{\sum_{i=1}^{\nplayer} a_i \cd \ndraw_i \cd (\ndraw_k + 2 \ndraw_j) - 2 \sum_{i=1}^{\nplayer} a_i \cd \ndraw_i^2}\\
0 &> \sum_{i=1}^{\nplayer} a_i \cd \ndraw_i \cd (\var \cd \ndraw_k + 2 \var \cd \ndraw_j -2\var \cd  \ndraw_i - \mue)
\end{align*}
Our goal is to show that this is negative if $\ndraw_i \leq \frac{\mue}{3\var}$ for all $i$.

First, we look over the portion of the sum equal to the $k$ index. This term is equal to: 
$$a_k \cd \ndraw_k \cd (2\var \cd \ndraw_j - \var \cd \ndraw_k - \mue)$$
which is negative, given our conditions. Next, we look at the $j$ term in the sum: 
$$a_j \cd \ndraw_j \cd (\var \cd \ndraw_k - \mue)$$
which is also negative. The remaining portions of the sum can be written as: 
$$(\total - a_j \cd \ndraw_j - a_k \cd \ndraw_k) \cd (\var \cd \ndraw_k + 2 \var \cd \ndraw_j - \mue) - 2 \var \sum_{i \ne j, k} a_i \cd \ndraw_i^2$$
which we would like to show is negative. We can maximize this term by holding $\total$ constant and minimizing the negative portion by setting $\ndraw_i=1$ for all other players besides $j, k$. This gives us an upper bound of: 
\begin{align*}
&\leq (\total - a_j \cd \ndraw_j - a_k \cd \ndraw_k) \cd (\var \cd \ndraw_k + 2 \var \cd \ndraw_j - \mue) - 2 \var (\total - a_j \cd \ndraw_j - a_k \cd \ndraw_k)\\
&= (\total - a_j \cd \ndraw_j - a_k \cd \ndraw_k) \cd (\var \cd \ndraw_k + 2 \var \cd \ndraw_j - \mue - 2 \var)
\end{align*}
Given the condition that $\ndraw_k, \ndraw_j \leq \frac{\mue}{3\var}$, we know that the coefficient is no more than 
$$3 \var \frac{\mue}{3\var} - \mue - 2 \var < 0$$
Taken together, this shows that the derivative of player $j$'s error with respect to $a_k$ is negative, which means that player $j$ always sees its error decrease with the addition of another player $k$. 
\end{proof}

\begin{lemma}\label{lem:case12}
Assume we have two groups of players, $A$ and $B$ with all players of size $\leq \frac{\mue}{3\var}$. Then, if either of the two conditions below are satisfied, the arrangement is not individually stable. 
\begin{enumerate}
    \item There exists $a \in A, b \in B$ such that $\ndraw_a = \ndraw_b$. 
    \item There exists $a \in A, b \in B$ such that $\ndraw_a > \ndraw_b$ and $\total_A - \ndraw_a \geq \total_B - \ndraw_b$. (Note that this could be defined symmetrically with respect to $B$). 
\end{enumerate}
\end{lemma}
\begin{proof}
First, we will assume that player $a$ does not wish to move to $B$ (if this is not true, then we already know that the arrangement is not IS). This tells us that: 
$$err_a(A) \leq err_a(B \cup \{\ndraw_a\})$$
Next, we will derive sufficient conditions for player $b$ to wish to move to $A$, or 
$$err_b(A \cup \{\ndraw_b\}) <err_b(B)$$
We will use the shorthand of $\total_A' = \total_A - \ndraw_a$ and $\total_B' = \total_B - \ndraw_b$. From the form of each player's error as in Lemma \ref{lem:error}, we can derive conditions for the difference in errors experienced by two players in the same coalition. Consider a coalition $\col$ and two players $j, k \in \col$, with $\ndraw_k \geq \ndraw_j$ Then, 
\begin{align*}
err_j(\col) - err_k(\col) &= \var \cd \frac{\sum_{i \ne j} \ndraw_i^2 +(\total_{\col} - \ndraw_j)^2}{\total_{\col}^2} -  \var \cd \frac{\sum_{i \ne k} \ndraw_i^2 + (\total_{\col} - \ndraw_k)^2}{\total_{\col}^2}\\
&= \var \cd \frac{\ndraw_k^2 - \ndraw_j^2 +(\total_{\col} - \ndraw_j)^2 - (\total_{\col} - \ndraw_k)^2}{\total_{\col}^2}\\
&=\var \cd \frac{\ndraw_k^2 - \ndraw_j^2 +(\total_{\col}^2 + \ndraw_j^2 - 2 \ndraw_j \cd \total_{\col}) - (\total_{\col}^2 + \ndraw_k^2 -2 \ndraw_k \cd \total_{\col})}{\total_{\col}^2}\\
&=\var \cd \frac{- 2 \ndraw_j \cd \total_{\col} +2 \ndraw_k \cd \total_{\col}}{\total_{\col}^2}\\
&= 2\var \cd \frac{\total_{\col} \cd (\ndraw_k - \ndraw_j)}{\total_{\col}^2} \\
&= 2 \var \cd \frac{\ndraw_k - \ndraw_j}{\total_{\col}}
\end{align*}
We can apply this derivation to obtain two equalities: 
$$err_b(A \cup b) = err_a(A \cup b) + 2 \var \frac{\ndraw_a - \ndraw_b}{\total_A' + \ndraw_a + \ndraw_b}$$
$$err_a(B \cup a) = err_b(B \cup a) - 2 \var \frac{\ndraw_a - \ndraw_b}{\total_B' + \ndraw_a + \ndraw_b}$$
So, rewriting the first inequality tells us that: 
$$err_a(A) \leq  err_b(B \cup a) - 2 \var \frac{\ndraw_a - \ndraw_b}{\total_B' + \ndraw_a + \ndraw_b}$$
Pulling over: 
$$err_a(A) + 2 \var \frac{\ndraw_a - \ndraw_b}{\total_B' + \ndraw_a + \ndraw_b}\leq err_b(B \cup a)$$
Note that because all of the players are of size $\leq \frac{\mue}{3\cd \var}$, we know by Lemma \ref{lem:welcome} that every player welcomes the addition of every other player, so 
$$err_b(B \cup a) < err_b(B)$$
In order to complete the proof, we need to show that $err_b(A \cup \{\ndraw_b\})$ is less than $err_a(A) + 2 \var \frac{\ndraw_a - \ndraw_b}{\total_B' + \ndraw_a + \ndraw_b}$. Again, because all of the players are of size $\leq \frac{\mue}{3\cd \var}$, we know from Lemma \ref{lem:welcome} that every player welcomes the addition of every other player, so
$$err_a(A \cup \{\ndraw_b\}) + 2 \var \frac{\ndraw_a - \ndraw_b}{\total_B' + \ndraw_a + \ndraw_b}<err_a(A) + 2 \var \frac{\ndraw_a - \ndraw_b}{\total_B' + \ndraw_a + \ndraw_b}$$
From our prior relation, we know that
$$err_b(A \cup b) -2 \var \frac{\ndraw_a - \ndraw_b}{\total_A' + \ndraw_a + \ndraw_b} +2 \var \frac{\ndraw_a - \ndraw_b}{\total_B' + \ndraw_a + \ndraw_b} =  err_a(A \cup \{\ndraw_b\}) + 2 \var \frac{\ndraw_a - \ndraw_b}{\total_B' + \ndraw_a + \ndraw_b}$$
Rewriting the term on the left tells us that what we want to show is:
$$err_b(A \cup b) \leq err_b(A \cup b) +2 \var \cd (\ndraw_a - \ndraw_b) \cd \p{\frac{1}{\total_B' + \ndraw_a + \ndraw_b} - \frac{1}{\total_A' + \ndraw_a + \ndraw_b}} $$
Now, we can apply our case analysis. If $\ndraw_a = \ndraw_b$, then the added coefficient is 0, so the final inequality holds. The inequality also holds if the fractional coefficient is positive or 0, or 
$$\frac{1}{\total_B' + \ndraw_a + \ndraw_b} \geq \frac{1}{\total_A' + \ndraw_a + \ndraw_b}$$
$$ \total_A' + \ndraw_a + \ndraw_b \geq \total_B' + \ndraw_a + \ndraw_b$$
$$ \total_A' \geq \total_B'$$
which is exactly the second criteria. 
\end{proof}

\begin{lemma}\label{lem:case3}
Assume we have two groups of players, $A$ and $B$, with all players of size $\leq \frac{\mue}{ 3\var}$. Define $\ndraw_a, \ndraw_b$ to be the largest players in $A, B$ respectively. Assume that $\ndraw_a > \ndraw_b$ and $\total_A -\ndraw_a < \total_B - \ndraw_b$, with  $\total_A - \ndraw_a \leq \frac{\mue}{3\var}$. Then, player $a$ would prefer to join $B$. 
\end{lemma}
\begin{proof}
We will show that the preconditions imply that player $a$ would wish to move to group $B$, or else
$$ err_a(B \cup \{\ndraw_a\}) < err_a(A) $$
Or, rewritten out, 
$$\frac{\mue}{\total_B + \ndraw_a} + \var \frac{\sum_{i \in B} \ndraw_i^2 + \total_B^2}{(\total_B + \ndraw_a)^2} < \frac{\mue}{\total_A} + \var \frac{\sum_{i \in A, i \ne a} \ndraw_i^2 + (\total_A - \ndraw_a)^2}{\total_A^2} $$
We will upper and lower bound the costs on both sides by taking the worst and best case scenario for how the $B$ and $A$ players can be arranged, respectively. We have already showed that we can minimize the total arrangement of fixed total mass by dividing it into players of size exactly 1, so
the player sizes equal to 1, or 
$$\sum_{i \in A, i \ne a} \ndraw_i^2 \geq  \total_A - \ndraw_a$$
Conversely, let's try to upper bound the $B$ sum. Previously, we did this by grouping all of the mass into a single player. In this case, we can't do this - we've assumed that the $\ndraw_b$ term is the largest of them, so the most we can set them to be equal to is $\ndraw_b$ exactly. However, the same reasoning still holds: if we keep the total $\total_B-\ndraw_b$ constant but rearrange them into groups of maximum size $b$, we only increase total cost. 
To see why, consider that we have $x, y$ with $x \geq y$, and some $x \leq b \leq x+y$. Then, we wish to show that: 

$$x^2 + y^2 < b^2 + (x+y-b)^2$$ 
Expanding: 
$$x^2 + y^2 < b^2 + (x+y-b)^2 = b^2 + b^2 + x^2 + y^2 -2b \cd x - 2 b\cd y +2 x \cd y$$
Cancelling common terms means we want to show: 
$$2b \cd x + 2 b\cd y  < 2b^2 +2 x \cd y$$
$$ x + y  < b +\frac{x \cd y}{b} < b + \frac{b \cd y}{b} = b + y$$
which is satisfied.

This result tells us that this process (grouping them into players of exactly size $\ndraw_b$, plus at most one player of size $< \ndraw_b$) does maximize the total sum, subject to this constraint. We will again use the shorthand of $\total_A' = \total_A -\ndraw_a$ and $\total_B' = \total_B - \ndraw_b$. Excluding player $\ndraw_b$, the mass is $\total_B'$, so the number of copies of $\ndraw_b$ that we can make is $\frac{\total_B'}{\ndraw_b}:= c + \epsilon$, for integer $c$ and $\epsilon \in [0, 1)$. If we know that $\epsilon = 0$ (which is always achievable), then we know that: 
$$\sum_{i \in B, i \ne b}\ndraw_i^2 \leq c \cd \ndraw_b^2 = \frac{\total_B'}{\ndraw_b} \cd \ndraw_b^2  = \total_B' \cd \ndraw_b$$
What if $\epsilon > 0$? Then, 
$$\sum_{i \in B, i \ne b}\ndraw_i^2 \leq c \cd \ndraw_b^2 + (\epsilon \cd \ndraw_b)^2 < c \cd \ndraw_b^2 + \epsilon \cd \ndraw_b^2 = \frac{\total_B'}{\ndraw_b} \cd \ndraw_b^2  = \total_B' \cd \ndraw_b$$
So, in either way, the $\total_B' \cd \ndraw_b$ term is an upper bound. This means that the worst-case scenario for us to show that: 
$$\frac{\mue}{\total_B' + \ndraw_a + \ndraw_b} + \var \frac{\total_B' \cd \ndraw_b + \ndraw_b^2 + (\total_B'+\ndraw_b)^2}{(\total_B' + \ndraw_a + \ndraw_b)^2} < \frac{\mue}{\total_A' + \ndraw_a} + \var \frac{\total_A'+ (\total_A')^2}{(\total_A' + \ndraw_a)^2} $$
We'll work by upper bounding the lefthand side. First, we'll replace the $\total_B'$. First, we'll also look at the derivative with respect to $\total_B'$, which gives: 
$$\frac{\ndraw_a (-\mue+3 \ndraw_b \cd \var+2 \total_b \cd \var)-(\ndraw_b+\total_B') (\mue+\ndraw_b \var)}{(\ndraw_a+\ndraw_b+\total_B')^3}$$
The numerator can be rewritten as: 
$$-\mue \cd (\ndraw_a + \ndraw_b + \total_B') - \var \cd \ndraw_b \cd (\total_B' + \ndraw_b) + 3 \var \cd \ndraw_a \cd \ndraw_b + 2 \var \cd \ndraw_a \cd \total_B'$$
We can show that this is negative because: 
$$\total_B' \cd (-\mue + 2 \var \cd \ndraw_a) < 0$$
since $\ndraw_a \leq \frac{\mue}{3\var}$. Similarly, 
$$\ndraw_b \cd (-\mue + 3 \var \cd \ndraw_a) \leq 0$$
Because the derivative with respect to $\total_B'$ is negative, we can over-bound it by setting it to its smallest value: $\total_A'+1$ (or $\total_A'$, for simplicity). This means that we can upper bound the lefthand side by writing:  
$$\frac{\mue}{\total_A' + \ndraw_a + \ndraw_b} + \var \frac{\total_A' \cd \ndraw_b + \ndraw_b^2 + (\total_A'+\ndraw_b)^2}{(\total_A' + \ndraw_a + \ndraw_b)^2} < \frac{\mue}{\total_A' + \ndraw_a} + \var \frac{\total_A'+ (\total_A')^2}{(\total_A' + \ndraw_a)^2} $$
Next, we'll work on replacing the $\ndraw_b$ term on the lefthand side. We start out by taking the derivative of the lefthand side with respect to $\ndraw_b$. This gives us: 
$$\frac{-\mue \cd (\ndraw_a +\total_A') + \var \cd \total_A'\cd (\total_A' + 3 \ndraw_a) + \ndraw_b \cd (-\mue + \var \cd (\total_A' + 4 \ndraw_a))}{(\total_A' + \ndraw_a + \ndraw_b)^3}$$
We will inspect the sign of the derivative, which is given by the numerator. Specifically, we will show that the derivative is always negative or 0 at $\ndraw_b = 0$, and is either negative forever, or else is negative and then positive. This implies that the lefthand side of the overall equation is either always decreasing in $\ndraw_b$ (implying that we can upper bound it by setting $\ndraw_b = 0$) or else is decreasing and then increasing (in which case the upper bound is either at $\ndraw_b =0$ or $\ndraw_b = \ndraw_a$). 

First, we will prove our claim about the derivative. At $\ndraw_b = 0$, the derivative is: 
$$- \mue \cd (\ndraw_a + \total_A') +\var \cd \total_A'\cd (\total_A' + 3 \ndraw_a) $$
We want to show this is negative, or: 
$$\var \cd \total_A'  \cd (3\ndraw_a  + \total_A') \leq  \mue \cd (\ndraw_a + \total_A') $$
Upper bounding the lefthand side: 
$$3 \var \cd \total_A'  \cd (\ndraw_a  + \total_A') \leq  \mue \cd (\ndraw_a + \total_A') $$
$$3 \var \cd \total_A' \leq  \mue  $$
which is satisfied by assumption. 
So, we know that the derivative starts out as 0 or negative. If the coefficient on $\ndraw_b$ (equal to $\var \cd (4\ndraw_a + \total_A') - \mue$) is negative, then the lefthand side of the overall equation is always decreasing as $\ndraw_b$ increases - so the upper bound at $\ndraw_b=0$ suffices. Otherwise, the curve is decreasing, then increasing. 

\textbf{Upper bound at $\ndraw_b=0$}\\
This bound is fairly straightforward. What we want to show is: 
$$\frac{\mue}{\total_A' + \ndraw_a} + \var \frac{ \total_A'^2}{(\total_A' + \ndraw_a)^2} < \frac{\mue}{\total_A' + \ndraw_a} + \var \frac{\total_A'+ (\total_A')^2}{(\total_A' + \ndraw_a)^2} $$
which is obviously true. 

\textbf{Upper bound at $\ndraw_b=\ndraw_a$}\\
This bound is trickier. (Note that technically, the upper bound is at $\ndraw_a -1$, but it is simpler to over-bound with $\ndraw_a$). What we'd like to show is: 
$$\frac{\mue}{\total_A' + 2\ndraw_a} + \var \frac{\total_A' \cd \ndraw_a + \ndraw_a^2 + (\total_A' + \ndraw_a)^2}{(\total_A' + 2\ndraw_a)^2} \leq \frac{\mue}{\total_A' + \ndraw_a} + \var \frac{\total_A' + \total_A'^2}{(\total_A' + \ndraw_a)^2}$$
We can write: 
$$\mue \cd \p{\frac{1}{\total_A' + \ndraw_a} - \frac{1}{\total_A' + 2 \ndraw_a}} = \mue \cd \frac{\ndraw_a}{(\total_A' + \ndraw_a) \cd (\total_A' + 2 \ndraw_a)}$$
Next, we move on to the $\var$ portion. Note that we can simplify the lefthand side, since: 
$$\total_A'^2 + \ndraw_a^2 + 2 \total_A' \cd \ndraw_a + \ndraw_a^2 + \total_A' \cd \ndraw_a = \total_A'^2 + 2\ndraw_a^2 + 3 \total_A' \cd \ndraw_a  = (\total_A' + \ndraw_a) \cd (\total_A' + 2 \ndraw_a)$$
So, the inequality we'd like to show becomes: 
$$\var \frac{\total_A' + \ndraw_a}{\total_A' + 2 \ndraw_a} - \var \frac{\total_A' + \total_A'^2}{(\total_A'+ \ndraw_a)^2} \leq \mue \cd \frac{\ndraw_a}{(\total_A' + \ndraw_a) \cd (\total_A' + 2 \ndraw_a)}$$
Simplifying the lefthand side gives: 
\begin{align*}
\var \cd \frac{(\total_A' + \ndraw_a)^3 - (\total_A' + \total_A'^2) \cd (\total_A' + 2 \ndraw_a)}{(\total_A' +2\ndraw_a) \cd (\total_A' + \ndraw_a)^2} &\leq \mue \cd \frac{\ndraw_a}{(\total_A' + \ndraw_a) \cd (\total_A' + 2 \ndraw_a)}\\
\var \cd \frac{(\total_A' + \ndraw_a)^3 - (\total_A' + \total_A'^2) \cd (\total_A' + 2 \ndraw_a)}{\total_A' + \ndraw_a} &\leq \mue \cd \ndraw_a
\end{align*}
We can make the lefthand side larger by making the negative part smaller - specifically, replacing the $(\total_A' + 2\ndraw_a)$ with a  $(\total_A' + \ndraw_a)$. This gives us: 
\begin{align*}
\var \cd \frac{(\total_A' + \ndraw_a)^3 - (\total_A' + \total_A'^2) \cd (\total_A' + \ndraw_a)}{\total_A' + \ndraw_a} &\leq \mue \cd \ndraw_a \\
\var \cd \p{(\total_A' + \ndraw_a)^2 - (\total_A' + \total_A'^2)} &\leq \mue \cd \ndraw_a
\end{align*}
Expanding out the lefthand side gives us: 
\begin{align*}
\var \cd (\total_A'^2 + \ndraw_a^2 + 2 \ndraw_a \cd \total_A' - \total_A' - \total_A'^2) &\leq \mue \cd \ndraw_a\\
\var \cd ( \ndraw_a^2 + 2 \ndraw_a \cd \total_A' - \total_A') &\leq \mue \cd \ndraw_a
\end{align*}
Again, we can make the lefthand side larger by dropping the negative portion: 
$$\var \cd ( \ndraw_a^2 + 2 \ndraw_a \cd \total_A') \leq \mue \cd \ndraw_a$$
$$\var \cd ( \ndraw_a + 2 \total_A') \leq \mue$$
Which is satisfied because we require $\total_A', \ndraw_a$ both $\leq \frac{\mue}{3\var}$. 
Note that, while this is a $\leq$, because we know that $\ndraw_b < \ndraw_a$, the overall inequality is strict. 
\end{proof}

\end{document}